\documentclass[10pt,journal,compsoc]{IEEEtran}

\usepackage{graphicx}
\usepackage[T1]{fontenc}
\usepackage{tikz}
\usepackage{amsmath, amssymb}
\usepackage{stmaryrd}
\usepackage{hyperref}
\usepackage{algorithm}
\usepackage{algpseudocode}
\usepackage{subcaption}
\usepackage[font={small,it}]{caption}
\usepackage{enumitem}
\usepackage{mathtools}
\usepackage{amsthm}

\usepackage{booktabs, multirow} 

\usepackage{pifont}

\ifCLASSOPTIONcompsoc
  \usepackage[nocompress]{cite}
\else
  \usepackage{cite}
\fi

\usetikzlibrary{automata, positioning, arrows}

\newcommand{\outputlang}[2]{\mathit{Out}_{#1}(#2)}

\newcommand{\hoaretriple}[3]{\{#1\}#2\{#3\}}
\usepackage{xspace}

\newcommand{\mypar}[1]{\vspace{1mm}\indent\textit{#1}}

\newtheorem{theorem}{Theorem}[section]
\newtheorem{lemma}[theorem]{Lemma}

\newcommand{\inputtype}{P}
\newcommand{\outputtype}{Q}

\newcommand{\outputbound}{l}
\newcommand{\distancebound}{d}

\newcommand{\DFAlanguage}[1]{\mathcal{L}(#1)}

\newcommand{\srcstate}{q_{\inputtype}} 
\newcommand{\trgstate}{q_{\outputtype}} 
\newcommand{\trnsstate}{q_T} 
\newcommand{\trnsstateTwo}{q'_T} 

\newcommand{\setcustomstate}[1]{Q_{#1}}
\newcommand{\setsrcstate}{Q_{\inputtype}} 
\newcommand{\settrgstate}{Q_{\outputtype}} 
\newcommand{\settrnsstate}{Q_T} 

\newcommand{\initstate}[1]{q^{\mathit{init}}_{#1}}
\newcommand{\finalstates}[1]{F_{#1}}

\newcommand{\transition}[1]{\delta_{#1}}
\newcommand{\deltaTstate}{\delta_{T}^{\mathit{st}}}
\newcommand{\deltaTout}{\delta_{T}^{\mathit{out}}}
\newcommand{\extendeddeltaTstate}{\delta_{T}^{\mathit{st*}}}

\newcommand{\extendeddeltaTout}{\delta_{T}^{\mathit{out*}}}
\newcommand{\extendedtransition}[1]{\delta_{#1}^{*}}

\newcommand{\vardeltastate}{\normalfont{\texttt{d}^\texttt{st}}}
\newcommand{\vardeltaout}{\normalfont{\texttt{d}^\texttt{out}_\texttt{ch}}}
\newcommand{\vardeltaoutlen}{\normalfont{\texttt{d}^\texttt{out}_\texttt{len}}}
\newcommand{\vareddist}{\normalfont{\texttt{ed}}}
\newcommand{\varprogress}{\normalfont{\texttt{en}}}
\newcommand{\varweight}{\normalfont{\texttt{wed}}}
\newcommand{\vardeltaoutputtype}{\normalfont{\texttt{d}_{\outputtype}}}

\newcommand{\invtransition}[1]{\delta_{#1}^{-1}}
\newcommand{\extendedinvtransition}[1]{\delta_{#1}^{-1*}}

\newcommand{\eddistfunc}[2]{\mathit{ed\!\_\!dist}(#1, #2)}
\newcommand{\meaneddist}[2]{\mathit{mean\!\_\!ed\!\_\!dist}(#1, #2)}
\newcommand{\energyfunc}[3]{\varprogress(#1, #2, #3)}

\newcommand{\lenfunc}[1]{\mathit{len}(#1)}

\newcommand{\inputexample}{s}
\newcommand{\outputexample}{t}
\newcommand{\examplemap}{{\inputexample \mapsto \outputexample}}
\newcommand{\exampleset}{[\overline{\inputexample \mapsto \outputexample}]}

\newcommand{\stringfont}[1]{\texttt{#1}}

\newcommand{\stringlength}[1]{\mathit{len}(#1)}

\newcommand{\achar}{\mathit{a}}

\newcommand{\position}{\normalfont{\texttt{config}}}

\newcommand{\triple}[3]{(#1, #2, #3)}
\newcommand{\simulationsym}{\normalfont{\texttt{sim}}}
\newcommand{\simulation}[3]{\simulationsym \triple{#1}{#2}{#3}} 
\newcommand{\lookaheadsimulationsym}{\simulationsym_R}
\newcommand{\lookaheadsimulation}[4]{\lookaheadsimulationsym(#1, #2, #3, #4)}

\newcommand{\sfa}{M}

\newcommand{\minterms}{\mathit{minterms}}
\newcommand{\mintermsshort}{\mathit{mterms}}
\newcommand{\finite}[1]{\mathit{F}(#1)} 
\newcommand{\mintermpred}[1]{mt(#1)} 

\newcommand{\intervalsize}{\mathit{size}}

\newcommand{\witnessshort}[1]{\mathit{wit}(#1)}

\newcommand{\lookaheadaut}{R}
\newcommand{\lookahead}{r} 
\newcommand{\lookaheadstate}{q_{\lookaheadaut}} 
\newcommand{\setlookaheadstate}{Q_{\lookaheadaut}}
\newcommand{\numlookaheadstates}{k_R}

\newcommand{\vardeltalookahead}{\normalfont{\texttt{d}_\texttt{\lookaheadaut}}}
\newcommand{\vardeltainverselookahead}{\normalfont{\texttt{d}^{-1}_\texttt{\lookaheadaut}}}
\newcommand{\varlookahead}{\normalfont{\texttt{look}}}

\newcommand{\bad}{T_{\textit{bad}}}

\newcommand{\badinput}{\inputtype_{\textit{bad}}}

\newcommand{\numsynthbenchmarks}{18}
\newcommand{\numrepairbenchmarks}{12}

\newtheorem{problemstatement}{Problem Statement}

\newcommand{\NA}{---}
\newcommand{\xmark}{\ding{55}\xspace}%

\def\name{\textsc{astra}\xspace}

\newcommand{\rone}{(\emph{i})~}
\newcommand{\rtwo}{(\emph{ii})~}
\newcommand{\rthree}{(\emph{iii})~}

\begin{document}

\title{Synthesizing Transducers from Complex Specifications}

\author{Anvay Grover,
        Ruediger Ehlers,
        and~Loris D'Antoni
\IEEEcompsocitemizethanks{\IEEEcompsocthanksitem A. Grover and L. D'Antoni are with the Univ. of Wisconsin-Madison.
\IEEEcompsocthanksitem R Ehlers is with Clausthal University of Technology.}
\thanks{Manuscript received April 19, 2005; revised August 26, 2015.}}

\markboth{Journal of \LaTeX\ Class Files,~Vol.~14, No.~8, August~2015}%
{Shell \MakeLowercase{\textit{et al.}}: Bare Demo of IEEEtran.cls for Computer Society Journals}

\IEEEtitleabstractindextext{%
\begin{abstract}
Automating string transformations has been one of the killer applications of program synthesis.
Existing synthesizers that solve this problem produce programs in domain-specific languages (DSL) that are engineered to help the synthesizer, and therefore lack nice formal properties.
This limitation prevents the synthesized programs from being used in verification applications (e.g., to check complex pre-post conditions) and makes the synthesizers hard to modify due to their reliance on the given DSL. 

We present a constraint-based approach to synthesizing transducers, a well-studied model with strong closure and decidability properties. Our approach handles three types of specifications: \rone input-output examples, \rtwo input-output types expressed as regular languages, and \rthree input/output distances that bound how many characters the transducer can modify when processing an input string.
Our work is the first to support such complex specifications and it does so by using the algorithmic properties of transducers to generate constraints that can be solved using off-the-shelf SMT solvers.
Our synthesis approach can be extended to many transducer models and it can be used, thanks to closure properties of transducers, to compute repairs for partially correct transducers.
\end{abstract} }

\maketitle

\IEEEdisplaynontitleabstractindextext

%
\IEEEpeerreviewmaketitle

\section{Introduction}
\label{sec:introduction}

String transformations are used in data transformations~\cite{gulwani2011flashfill}, sanitization of untrusted inputs~\cite{Bek2011,LorisVMCAI2013}, and many other domains~\cite{programmable-string-robustness}. 
Because in these domains bugs may cause serious security vulnerabilities~\cite{Bek2011}, there has been increased interest
in building tools that can help programmers verify~\cite{Bek2011,LorisVMCAI2013} and synthesize~\cite{gulwani2011flashfill,symmetriclenses,bijectivelenses} string transformations.

Techniques for \textit{verifying}  string transformations rely on automata-theoretic approaches that provide  powerful decidability properties~\cite{Bek2011}.
On the other hand, techniques for \textit{synthesizing} string transformations rely on domain-specific languages (DSLs)~\cite{gulwani2011flashfill,symmetriclenses}.
These DSLs are designed to make synthesis practical and have to give up the closure and decidability properties enabled by automata-theoretic models. 
The disconnect between the two approaches raises a natural question: \textit{Can one synthesize automata-based models and therefore retain and leverage their elegant properties?}

Transducers are a natural automata-based formal model for synthesizing string transformations.
A \textit{finite state transducer} (FT) is an automaton where each transition reads an input
character and outputs a string of output characters.
For instance, Figure~\ref{fig:escape_quotes} shows a transducer that `escapes' instances of the \texttt{"} character.
So, on input \texttt{a"\textbackslash "a}, the transducer outputs the string \texttt{a\textbackslash "\textbackslash \textbackslash "a}.
FTs have found wide adoption in a variety of domains~\cite{mohri1997finite,LorisVMCAI2013} because of their many desirable properties (e.g., decidable equivalence check and closure under composition~\cite{cacmSFA}).
There has been increasing work on building SMT solvers for strings that support transducers; the Ostrich tool \cite{ostrichATVA2020} allows a user to write programs in SMT where string-transformations are modelled using transducers.
One can then write constraints over such programs and use an SMT solver to automatically check for satisfiability or prove unsatisfiability of those constraints.
For example, given a program like the following:
\begin{verbatim}
y = escapeQuotes(x)
z = escapeQuotes(y)
assert(y==z) //Checking idempotence
\end{verbatim}
one can use Ostrich to write a set of constraints and use them to prove whether the assertion holds.
However, to do so, one needs to first write a transducer $T$ that implements the function \verb|escapeQuotes|.
However, writing transducers by hand is a cumbersome and error-prone task.

\begin{figure}[tb]
  \centering
  \begin{subfigure}{0.5\textwidth}
  \centering
    \resizebox{0.44\textwidth}{0.24\textwidth}{
  \begin{tikzpicture}[shorten >=1pt, node distance=2.5cm, auto] 
    \node[state, accepting, initial] (q0) {$q_0$};
    \node[state, accepting, right of=q0] (q1) {$q_1$};
    \path[->]  
            (q0) edge[loop above] node{\texttt{a} $\rightarrow$ \texttt{a}} (q0)
            
            (q0) edge[loop below] node{\texttt{"} $\rightarrow$ \texttt{\textbackslash "}} (q0)
            
            (q0) edge node[above] {\texttt{\textbackslash} $\rightarrow$ \texttt{\textbackslash}} (q1)
            
            (q1) edge[bend left=55] node{\texttt{a} $\rightarrow$ \texttt{a}} (q0)
            
            (q1) edge[bend left=20] node{\texttt{"} $\rightarrow$ \texttt{"}} (q0)
            
            (q1) edge[bend right=45, above] node {\texttt{\textbackslash} $\rightarrow$ \texttt{\textbackslash}} (q0)
            ;
    \end{tikzpicture}
    }
    \caption{Transducer EscapeQuotes}
    \label{fig:escape_finite_correct}
    \end{subfigure}
    \begin{subfigure}{0.5\textwidth}
    \vspace{1mm}
    \centering
    \fbox{%
    \scriptsize
        \parbox{0.95\textwidth}{%
            \textbf{Examples:}  
            \{\texttt{a"a} $\mapsto$ \texttt{a\textbackslash"a}, \texttt{a\textbackslash \textbackslash a} $\mapsto$ \texttt{a\textbackslash \textbackslash a}, \texttt{a\textbackslash a} $\mapsto$ \texttt{a\textbackslash a}, \texttt{a\textbackslash "a} $\mapsto$ \texttt{a\textbackslash "a}, \texttt{\textbackslash} $\mapsto$ \texttt{\textbackslash}\} \\
            \textbf{Types:} 
            \texttt{[a"]$^\ast$\textbackslash?|([a"]$^\ast$\textbackslash[a"\textbackslash][a"]$^\ast$)$^\ast$} $\to$ 
            \texttt{a$^\ast$\textbackslash?|(a$^\ast$\textbackslash[a"\textbackslash]a$^\ast$)$^\ast$} \\
            \textbf{Distance:} At most 1 edit per input character
        }%
    }
    \caption{Specification to synthesize EscapeQuotes}
    \label{fig:specification_box}
    \end{subfigure}
\caption{Simplified version of EscapeQuotes from~\cite{Bek2011}.}
\label{fig:escape_quotes}    
\end{figure}

In this paper, we present a technique for synthesizing transducers from high-level specifications.
We use three different specification mechanisms to quickly yield desirable transducers: input-output examples, input-output types, and input-output distances.
When provided with the specification in Figure~\ref{fig:specification_box}, our approach yields the  transducer in Figure~\ref{fig:escape_quotes}.
While none of the three specification mechanisms are effective in isolation, they work well altogether.
Input-output examples are a natural specification mechanism that is easy to provide, but only capture finitely many inputs.
Similarly, input-output types are a natural way to prevent a transducer from generating undesired strings and can often be obtained from function/API specifications.
Last, input-output distances are a natural way to specify how much of the input string should be preserved by the transformation.

We show that if the size of the transducers is fixed, all such specifications can be encoded as a set of constraints 
whose solution directly provides a transducer. 
While the constraints for examples are fairly straightforward, we introduce two new ideas for encoding types and distances.
To encode types and distances, we show that one can use constraints to ``guess'' the simulation relation and the invariants necessary to prove that
the transducer has the given type and respects the given distance constraint.

Because our constraint-based approach is based on decision procedures and is  modular, it can support more complex models of transducers:
\rone Symbolic Finite Transducers (s-FTs), which support large alphabets~\cite{LorisCAV17}, and
\rtwo FTs with lookahead, which can express functions that otherwise require non-determinism.
In addition, closure properties of transducers allow us to reduce repair problems for string transformations to our synthesis problem.

\noindent \vspace{1mm}\textit{Contributions:}
We make the following contributions.
\begin{itemize}[topsep=2pt]
    \item A constraint-based synthesis algorithm for synthesizing transducers from complex specifications (Sec.~\ref{sec:encoding}).
    \item Extensions of our synthesis algorithm to more complex models---e.g., symbolic transducers and transducers with lookahead---and problems---e.g., transducer repair---that showcase the flexibility of our approach and the power of working with transducers, which enjoy strong theoretical properties---unlike domain-specific languages (Sec.~\ref{sec:extensions}).
    \item \name: a tool that can synthesize and repair transducers and compares well with a state-of-the-art tool for synthesizing string transformations (Sec.~\ref{sec:eval}).
\end{itemize}

\noindent Proofs and additional results are available in the appendix.

\section{Transducer Synthesis Problem}
\label{sec:transducer-synthesis-problem}

In this section, we define the transducer synthesis problem.

A \textit{deterministic finite automaton} (DFA) over an alphabet $\Sigma$ is a tuple $D = (\setcustomstate{D}, \transition{D}, \initstate{D}, \finalstates{D})$: $\setcustomstate{D}$ is the set of states, $\transition{D}: \setcustomstate{D} \times \Sigma \rightarrow \setcustomstate{D}$ is the transition function, $\initstate{D}$ is the initial state, and $\finalstates{D}$ is the set of final states.
The extended transition function $\extendedtransition{D} : \setcustomstate{D} \times \Sigma^\ast \rightarrow \setcustomstate{D}$ is defined as $\extendedtransition{D}(q, \varepsilon) = q$ and $\extendedtransition{D}(q, au) = \extendedtransition{D}(\transition{D}(q, a), u)$.
We say that $D$ accepts a string $w$ if $\extendedtransition{D}(\initstate{D}, w) \in \finalstates{D}$. 
The \textit{regular language} $\mathcal{L}(D)$ is the set of strings accepted by a DFA $D$.
%

A total \textit{finite state transducer} (FT) is a tuple $T = (\setcustomstate{T}, \deltaTstate, \deltaTout, \initstate{T})$,
where
$\setcustomstate{T}$ are states and $\initstate{T}$ is the initial state. 
Transducers have two transition functions:
$\deltaTstate: \trnsstate \times \Sigma \rightarrow \trnsstate$ defines the target state, while $\deltaTout: \trnsstate \times \Sigma \rightarrow \Sigma^\ast$
 defines the output string of each transition.
The extended function for states $\extendeddeltaTstate$ is defined analogously to the extended transition function for a DFA. 
The extended function for output strings is defined as
$\extendeddeltaTout(q, \varepsilon) =\varepsilon$ and $\extendeddeltaTout(q, au)= \extendeddeltaTout(q, a)\cdot \deltaTout(\extendeddeltaTstate(q, a), u)$.
Given a string $w$ we use $T(w)$ to denote $\extendeddeltaTout(\initstate{T}, w)$, i.e., the output string generated by $T$ on $w$.
Given two DFAs $\inputtype$ and $\outputtype$, we write $\hoaretriple{\inputtype}{T}{\outputtype}$ for a transducer $T$ iff for every string $\inputexample$ in $\DFAlanguage{\inputtype}$, the output string $T(\inputexample)$ belongs to $\DFAlanguage{\outputtype}$.

An \textit{edit operation} on a string is either an insertion/deletion of a character, or a replacement of a character with a different one.
For example, editing the string \stringfont{ab} to the string \stringfont{acb} requires one edit operation, which is inserting a \stringfont{c} after the \stringfont{a}.
The \textit{edit distance} $\eddistfunc{\inputexample}{\outputexample}$ between two strings $\inputexample$ and $\outputexample$ is the number of edit-operations required to reach $\outputexample$ from $\inputexample$. 
We use $\lenfunc{w}$ to denote the length of a string $w$.
The \textit{mean edit distance} $\meaneddist{\inputexample}{\outputexample}$ between two strings $\inputexample$ and $\outputexample$ is defined as $\eddistfunc{\inputexample}{\outputexample} / \textit{len}(\inputexample)$. 
For example, the mean edit distance from \stringfont{ab} to \stringfont{acb} is $1/2=.5$.

We can now formulate the transducer synthesis problem.
We assume a fixed alphabet $\Sigma$.
If the specification requires that $\inputexample$ is translated to $\outputexample$, we write that as $\examplemap$.

\begin{problemstatement}[Transducer Synthesis] \sloppypar The transducer synthesis problem has the following inputs and output:\newline
\textbf{Inputs}
\begin{itemize}[topsep=1pt]
    \item Number of states $k$ and upper bound $\outputbound$ on the length of the output of each transition.
    
    \item Set of input-output examples $E=\exampleset$. 
    
    \item Input-output types $\inputtype$ and $\outputtype$, given as DFAs.
    
    \item A positive upper bound $\distancebound\in\mathbb{Q}$  on the mean edit distance.
\end{itemize}
\textbf{Output} A total transducer $T = (\setcustomstate{T}, \deltaTstate, \deltaTout, \initstate{T})$ with $k$ states such that:
\begin{itemize}[topsep=1pt]
\item Every transition of $T$ has an output with length at most $\outputbound$, i.e.,  $\forall \trnsstate \in \settrnsstate, \achar \in \Sigma.\ \lenfunc{\deltaTout(q, a)} \leq \outputbound$.
\item $T$ is consistent with the examples: $\forall \examplemap\in E.\  T(\inputexample) = \outputexample$.
\item $T$ is consistent with input-output types, i.e., $\hoaretriple{\inputtype}{T}{\outputtype}$.
\item For every string $w \in\inputtype$, $\meaneddist{w}{T(w)} \leq \distancebound$.
\end{itemize} 
\end{problemstatement}

The synthesis problem that we present here is for FTs, and in Section~\ref{sec:encoding}, we provide a sound algorithm to solve it using a system of constraints.
One of our key contributions is that 
our encoding can be easily adapted to synthesizing richer models than FTs (e.g., symbolic transducers~\cite{cacmSFA} and transducers with regular lookahead), while still using the same encoding building blocks (Section~\ref{sec:extensions}).

\section{Constraint-based Transducer Synthesis}
\label{sec:encoding}

In this section, we present a way to generate constraints to solve the transducer synthesis problem  defined in Section~\ref{sec:transducer-synthesis-problem}.
The synthesis problem can then be solved by invoking a \textit{Satisfiability Modulo Theories} (SMT) solver on the constraints.

We use a constraint encoding, rather than a direct algorithmic approach because of the multiple objectives to be satisfied.
Synthesizing a transducer that translates a set of input-output examples is already an NP-Complete problem~\cite{minimaltransducersynthesis}.
On top of that, we also need to handle input-output types and distances.
Our encoding is divided into three parts, one for each objective, which are presented in the following subsections.
This division makes our encoding very modular and programmable. In Section~\ref{sec:extensions} we show how it can be adapted to different transducer models and problems.
We include a brief description of the size of the constraint encoding in Section~\ref{sec:size_encoding} of the appendix.

\subsection{Representing Transducers}
\label{sec:representingTransducers}

\looseness-1 The transducer we are synthesizing has $k$ (part of the problem input) states $Q_T = \allowbreak{} \{q_0, ..., \allowbreak{} q_{k-1}\}$. 
%
We often use $\initstate{T}$ as an alternative for $q_0$, the initial state of $T$.

We illustrate how a transition $q_1 \xrightarrow{\texttt{a}/\texttt{bc}} q_2$ is represented in our encoding. 
The target state is captured using an uninterpreted function $\vardeltastate:Q_T\times \Sigma \to Q_T$, e.g., $\vardeltastate(q_1,\texttt{a}) = q_2$.
Representing the output of the transition is trickier because its length is not known a priori.
The synthesis problem however provides an output bound $\outputbound$, which allows us to limit the number of characters that may appear in the output.
We use an uninterpreted function $\vardeltaout:Q_T\times \Sigma \times \{0,\ldots,\outputbound{-}1\} \to \Sigma$
to represent each character in the output string; in our example,
$\vardeltaout(q_1,\texttt{a},0) = \texttt{b}$ and
$\vardeltaout(q_1,\texttt{a},1) = \texttt{c}$.
Since an output string's length can be smaller than $\outputbound$, we use an additional uninterpreted function $\vardeltaoutlen{}:Q_T\times \Sigma \to \{0,\ldots,\outputbound\}$ to model the length of a transition's output; in our example
$\vardeltaoutlen(q_1,\texttt{a})=2$.

We say an assignment to the above variables extends to a transducer $T$ for the transducer $T$ obtained by instantiating $\delta^{st}$ and $\delta^{out}$ as described above.

\subsection{Input-output Examples}
\fbox{%
    \parbox{0.47\textwidth}{%
       \textbf{Goal:} For each input output-example $\examplemap\in E$, $T$ should translate $\inputexample$ to $\outputexample$. 
    }%
}
\vspace{1ex}



Translating $\inputexample$ to the correct output string means that $\extendeddeltaTout(\initstate{T}, \inputexample) = \outputexample$.
Generating constraints that capture this behavior of $T$ on an example is challenging because we do not know a priori what parts of $\outputexample$ are produced by what steps of the transducer's run.
Suppose that we need to translate $\inputexample=a_0a_1$ to $\outputexample=b_0b_1b_2$. 
A possible solution is for the transducer to have the run $q_0 \xrightarrow{a_0/b_0} q1 \xrightarrow{a_1/b_1b_2} q_2$.
Another possible solution might be to instead have $q_0 \xrightarrow{a_0/b_0b_1} q1 \xrightarrow{a_1/b_2} q_2$.
Notice that the two runs traverse the same states but produce different parts of the output strings at each step.
Intuitively, we need a way to ``track'' how much output the transducer has produced before processing the $i$-th character in the input and what state it has landed in.
For every input example $\examplemap$ such that $\inputexample=a_0\cdots a_n$ and $\outputexample=b_0\cdots b_m$, we introduce an uninterpreted function $\position_\inputexample: \{0,\ldots,n\} \rightarrow \{0,\ldots,m\} \times \settrnsstate$ such that
$\position_\inputexample(i)=(j,\trnsstate)$ iff after reading $a_0\cdots a_{i-1}$, the transducer $T$ has produced the output $b_0\cdots b_{j-1}$ and reached state $\trnsstate$---i.e., $\extendeddeltaTout(q_0,a_0\cdots a_{i-1})=b_0\cdots b_{j-1}$ and $\extendeddeltaTstate(q_0,a_0\cdots a_{i-1}) = \trnsstate$.

We describe the constraints that describe the behavior of $\position_\inputexample$.
Constraint~\ref{eq:examples1} states that a configuration must start at the initial state and be at position 0 in the output.
\begin{equation}
\label{eq:examples1}
\position_\inputexample(0) = (0, \initstate{T})
\end{equation}

Constraint~\ref{eq:examples2} captures how the configuration is updated when reading the $i$-th character of the input.
For every $0\leq i < n$, $0 \leq j < m$, $c\in\Sigma$, and $q_T\in Q_T$:
\begin{equation}
\label{eq:examples2}
\begin{split}
    & \position_\inputexample(i) = (j, \trnsstate)
    \wedge
    a_i = c
    \Rightarrow \\
    & 
    \quad\quad[
    \bigwedge_{0 \leq z < l} (\vardeltaout(\trnsstate, c,z) = b_{j+z} \vee z\geq  \vardeltaoutlen(\trnsstate,c)) \wedge\\
    &
    \quad\quad\position_\inputexample(i + 1) = (j + \vardeltaoutlen(\trnsstate,c), \vardeltastate(\trnsstate, c))
    ] 
\end{split}
\end{equation}

Informally, if the $i$-th character is $c$ and the transducer has reached state $\trnsstate$ and produced the characters $b_0\cdots b_{j-1}$ so far, the transition reading $c$ from state $q_T$ outputs characters $b_{j}\cdots b_{j+f-1}$, where $f$ is the output length of the transition.
The next configuration is then $(j+f, \vardeltastate(\trnsstate, c))$.

Finally, Constraint~\ref{eq:examples3} forces $T$ to be completely done with generating $\outputexample$ when $\inputexample$ has been entirely read.
Recall that $\stringlength{\inputexample} = n$ and $\stringlength{\outputexample} = m$. 
\begin{equation}
\label{eq:examples3}
  \bigvee_{\trnsstate \in \settrnsstate} \position_\inputexample(n)
  = (m, \trnsstate)
\end{equation}

\noindent The constraint encoding for examples is sound and complete (Proofs in~\ref{sec:examples_proofs}).

\subsection{Input-Output Types}

\fbox{%
    \parbox{0.47\textwidth}{%
       \textbf{Goal:} $T$ should satisfy the property  $\hoaretriple{\inputtype}{T}{\outputtype}$.
    }%
}
\vspace{1ex}

Encoding this property using constraints is challenging because it requires enforcing
that when $T$ reads one of the (potentially) infinitely many strings in $\inputtype$ it always outputs
a string in $\outputtype$.
To solve this problem, we draw inspiration from how one proves that the property
$\hoaretriple{\inputtype}{T}{\outputtype}$ holds---i.e., using a simulation relation that relates runs over $\inputtype$, $T$ and $\outputtype$.
Intuitively, if $\inputtype$ has read some string $w$, we need to be able to encode the behavior of $T$ in terms 
of $w$, i.e., what state of $T$ this transducer is in after reading $w$ and what output string $w'$ it produced.
Further, we also need to be able to encode in which state $\outputtype$ would be after reading the output string $w'$.
We do this by introducing a function \simulationsym{}: $\setsrcstate \times \settrnsstate \times \settrgstate \rightarrow \{0, 1\}$, 
which preserves the following invariant: $\simulation{\srcstate}{\trnsstate}{\trgstate}$ holds if there exist strings $w,w'$ such that
    $\extendedtransition{\inputtype}(\initstate{\inputtype}, w) = q_{\inputtype}$, 
    $\extendeddeltaTstate(\initstate{T}, w) = \trnsstate$, 
    $\extendeddeltaTout(\initstate{T}, w) = w'$, 
    and $\extendedtransition{\outputtype}(\initstate{\outputtype}, w') = q_{\outputtype}$.
%

Constraint~\ref{eq:types1} states the initial condition of the simulation---i.e., $P$, $T$, and $Q$ are in their initial states.
\begin{equation} 
\label{eq:types1}
    \simulation{\initstate{\inputtype}}{\initstate{T}}{\initstate{\outputtype}}
\end{equation}

Constraint~\ref{eq:types2} encodes how we advance the simulation relation for states $\srcstate,\trnsstate,\trgstate$ and for a character $c\in\Sigma$, using free variables $c_0 \ldots, c_{l-1}$ and $\trgstate^0 \ldots, \allowbreak{} \trgstate^{l}$ that are separate for each combination of $\srcstate,\trnsstate,\trgstate$, and $c$:
\begin{equation} 
\label{eq:types2}
    \begin{split}
    \simulation{\srcstate}{\trnsstate}{\trgstate} \Rightarrow & 
    \bigwedge_{\mathclap{0 \leq z \leq \outputbound}}  (\vardeltaoutlen(\trnsstate, c) = z \Rightarrow \\
    & \quad [\bigwedge_{\mathclap{0 \leq x < z}} \vardeltaout(\trnsstate, c, x) {=} c_x] \wedge \\
    & \quad [\trgstate^0 {=} \trgstate \wedge \bigwedge_{\mathclap{1 \leq x < z}} \trgstate^x {=} 
    \vardeltaoutputtype(\trgstate^{x-1},c_{x-1})]
    \wedge \\
    & \quad \simulation{\transition{\inputtype}(\srcstate, \mathit{c})}{\vardeltastate(\trnsstate, c)}{\trgstate^z})
    \end{split}
\end{equation}

Intuitively, if $\simulation{\srcstate}{\trnsstate}{\trgstate}$ and we read a character $c$, $\inputtype$ moves to $\transition{\inputtype}(\srcstate, \mathit{c})$
and $T$ moves to $\vardeltastate{}(\srcstate, \mathit{c})$. 
However, we also need to advance $\outputtype$ and the $\vardeltaoutlen{}$ symbols produced by $\vardeltaout{}$.
We hard-code the transition relation $\transition{\outputtype}$ in an uninterpreted function $\vardeltaoutputtype: \settrgstate \times \Sigma \rightarrow \settrgstate$, and apply it to compute the output state reached when reading the output string.
E.g., if $\vardeltaoutlen{}(\trnsstate, c) = 2$ and $\vardeltaout(\trnsstate, c, 0)=c_0$ and $\vardeltaout(\trnsstate, c, 1)=c_1$, the next state in $\outputtype$ is $\vardeltaoutputtype(\vardeltaoutputtype(\trgstate, c_0), c_1)$.

Lastly, Constraint~\ref{eq:types3} states that if we encounter a string in $\DFAlanguage{\inputtype}$---i.e., $\inputtype$ is in a state $\srcstate \in \finalstates{\inputtype}$---the relation does not contain a state $\trgstate \notin \finalstates{\outputtype}$. Since  $\outputtype$ is deterministic, this means that $\outputtype$ accepts $T$'s output.
\begin{equation}
\label{eq:types3} 
    \bigwedge_{\srcstate \in \finalstates{\inputtype}}
    \bigwedge_{\trgstate \notin \finalstates{\outputtype}}
    \neg \simulation{\srcstate}{\trnsstate}{\trgstate}
\end{equation}

\noindent Soundness and completeness of the type constraints are proven in~\ref{sec:types_proofs}.

\subsection{Input-output Distance}

\fbox{%
    \parbox{0.47\textwidth}{%
       \textbf{Goal:} The mean edit distance between any input string $w$ in $\DFAlanguage{\inputtype}$ and the output string $T(w)$ should not exceed $\distancebound$.
    }%
}

\vspace{1mm}
Capturing the edit distance for all the possible inputs in the language of $P$ and the corresponding  outputs produced by the transducer 
is challenging because these sets can be infinite.
Furthermore, exactly computing the edit distance between an input and an output string may involve comparing characters appearing on different transitions in the transducer run.
For example, consider the transducer shown in Figure~\ref{fig:distance_example1} and suppose that we are only interested in strings in the input type $\inputtype = \texttt{a(ba)}{\ast}\texttt{a}$.
The first transition from $q_0$ deletes the \texttt{a}, therefore making 1 edit.
This transducer has a cycle between states $q_1$ and $q_2$, which can be taken any number of times.
Each iteration, locally, would require that we make 2 edits: one to change the \texttt{b} to \texttt{a}, and the other to change the \texttt{a} to \texttt{b}.
However, the total number of edits made over any string in the input type $\inputtype = \texttt{a(ab)}{\ast}\texttt{a}$ by this transducer is 1, because the transducer changes strings of the form $\texttt{a(ba)}^n\texttt{a}$ to be of the form $\texttt{(ab)}^n\texttt{a}$.
Looking at the transitions in isolation, we are prevented from deducing that the edit distance is always 1 because the first transition delays outputting a character.
If there was no such delay, as is the case for the transducer in Figure~\ref{fig:distance_example2}, which is equivalent on the relevant input type to the one in Figure~\ref{fig:distance_example1}, then this issue would not arise.

We take inspiration from Benedikt et al.~\cite{Benedikt11} and focus on the simpler problem of synthesizing a transducer that has `aggregate cost' that satisfies the given objective.\footnote{Benedikt et al.~\cite{Benedikt11} studied a variant of the problem where the distance is bounded by some finite constant.
%
%
Their work shows that when there is a transducer between two languages that has some bounded global edit distance, then there is also a transducer that is bounded (but with a different bound) under a local method of computing the edit distance---i.e., one where the computation of the edit distance is done transition by transition.}
For a transducer $T$ and string $s = a_0\ldots a_n$, let $\initstate{T} \xrightarrow{a_0 / y_0} \trnsstate^1 
\ldots
\trnsstate^{n} \xrightarrow{a_n / y_n} \trnsstate^{n+1}$ be the run of $s$ on $T$.
Then, the \textit{aggregate cost} of $T$ on $s$ is the sum of the edit distances $\eddistfunc{a_i}{y_i}$ over all indices $0 \leq i \leq n$.
The \textit{mean aggregate cost} of $T$ on $s$ is the aggregate cost divided by $\lenfunc{s}$, the length of $s$.
It follows that if $T$ has a mean aggregate cost lower than some specified $d$ for every string, then it also has a mean edit distance lower than $d$ for every string.

However, the mean aggregate cost overapproximates the edit distance, e.g., the transducer in Figure~\ref{fig:distance_example1} has mean aggregate cost 1, while the mean edit distance when considering only strings in $\inputtype = \texttt{a(ab)}{\ast}\texttt{a}$ is less than $1/2$.
For this reason, if the mean edit distance objective was set to $1/2$, our constraint encoding can only synthesize the transducer in Figure~\ref{fig:distance_example2}, and not the equivalent one in Figure~\ref{fig:distance_example1}.

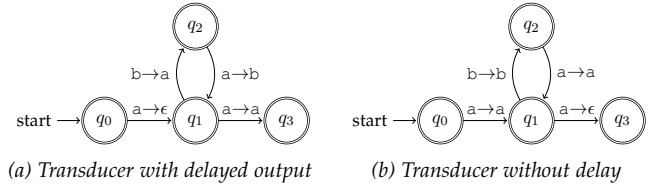
\begin{figure}[t]
\captionsetup[subfigure]{justification=centering}
  \centering
  \begin{subfigure}{0.23\textwidth}
  \centering
  \resizebox{0.98\textwidth}{0.45\textwidth}{
    \begin{tikzpicture}[shorten >=1pt, node distance=1.75cm, auto]
    \node[state, accepting, initial] (q0) {$q_0$};
    \node[state, accepting, right of=q0] (q1) {$q_1$};
    \node[state, accepting, above of=q1] (q2) {$q_2$};
    \node[state, accepting, right of=q1] (q3) {$q_3$};
    
    \path[->] 
    (q0) edge node {$\texttt{a} {\rightarrow} \epsilon$} (q1)
    
    (q1) edge[bend left=30] node {$\texttt{b} {\rightarrow} \texttt{a}$} (q2)
    
    (q2) edge[bend left=30] node {$\texttt{a} {\rightarrow} \texttt{b}$} (q1)
    
    (q1) edge node {$\texttt{a} {\rightarrow} \texttt{a}$} (q3)
    ;
  \end{tikzpicture}
  }
  \caption{Transducer with delayed output}
  \label{fig:distance_example1}
  \end{subfigure}
  ~
  \begin{subfigure}{0.23\textwidth}
  \centering
  \resizebox{0.98\textwidth}{0.45\textwidth}{
    \begin{tikzpicture}[shorten >=1pt, node distance=1.75cm, auto]
    \node[state, accepting, initial] (q0) {$q_0$};
    \node[state, accepting, right of=q0] (q1) {$q_1$};
    \node[state, accepting, above of=q1] (q2) {$q_2$};
    \node[state, accepting, right of=q1] (q3) {$q_3$};
    
    \path[->] 
    (q0) edge node {$\texttt{a} {\rightarrow} \texttt{a}$} (q1)
    
    (q1) edge[bend left=30] node {$\texttt{b} {\rightarrow} \texttt{b}$} (q2)
    
    (q2) edge[bend left=30] node {$\texttt{a} {\rightarrow} \texttt{a}$} (q1)
    
    (q1) edge node {$\texttt{a} {\rightarrow} \epsilon$} (q3)
    ;
  \end{tikzpicture}
  }
  \caption{Transducer without delay}
  \label{fig:distance_example2}
  \end{subfigure}
  \caption{Transducers with and without delay.}
\end{figure}

Our encoding is complete for transducers in which the aggregate cost coincides with the actual edit distance.
We leave the problem of being complete with regards to global edit distance as an open problem. 
In fact, we are not even aware of an algorithm for \textit{checking} (instead of synthesizing) whether a transducer satisfies a mean edit distance objective.\footnote{The mean edit distance is similar to mean payoff~\cite{DBLP:reference/mc/BloemCJ18}, which discounts a cost by the length of a string and looks at the behavior of a transducer in the limit. Our distance is different because 1) it looks at finite-length strings, and 2) it requires computing the edit distance, which cannot be done one transition at a time.}
In Section~\ref{sec:regular-look-ahead}, we present transducers with lookahead, which can mitigate this source of incompleteness.
Furthermore, our evaluation shows that using the aggregate cost and enabling lookahead are both effective techniques in practice.

We can now present our constraints.
First, we provide constraints for the edit distance of individual
transitions (recall that transitions are being synthesized and we therefore
need to compute their edit distances separately).
Secondly, we provide constraints that implicitly compute state invariants to capture the aggregate cost between 
input and output strings at various points in the computation.
We are given a rational number $\distancebound$ as an input to the problem, which is the allowed distance bound.

\mypar{Edit Distance of Individual Transitions.}
To compute the edit distance between the input and the output of each transition, we introduce a function \vareddist{}: $\settrnsstate \times \Sigma \rightarrow \mathbb{Z}$.
For a transition from state $\trnsstate$ reading a character $c$,  $\vareddist{}(\trnsstate,c)$ represents the edit distance between $c$ and $\deltaTout(\trnsstate, c)$. Notice that this quantity is bounded by the output bound $\outputbound$.
The constraints to encode the value of this function are divided into two cases:
i) the output of the transition contains the input character $c$ (Constraint \ref{eq:distance1}),
ii) the output of the transition \textit{does not} contain the input character $c$ (Constraint \ref{eq:distance2}).
In both cases, the values are set via a simple case analysis on whether the length of the output is 0 (edit distance is 1) or not (the edit distance is related to the length of the output).
\begin{equation} 
\label{eq:distance1}
\begin{split}
    & [\bigvee_{\mathclap{0 \leq z < \vardeltaoutlen(\trnsstate, c)}} \vardeltaout(\trnsstate, c, z) = c] 
    \Rightarrow \\
    & \quad [
    \vardeltaoutlen(\trnsstate, c) = 0 \Rightarrow
    \vareddist(\trnsstate, c) = 1  \wedge \\
    & \quad \vardeltaoutlen(\trnsstate, c) \neq 0 \Rightarrow
    \vareddist(\trnsstate, c) = \vardeltaoutlen(\trnsstate, c) - 1 ]
\end{split}
\end{equation}
\begin{equation}
\label{eq:distance2}
\begin{split}
    & [\bigwedge_{\mathclap{0 \leq z < \vardeltaoutlen(\trnsstate, c)}}  \vardeltaout(\trnsstate, c, z) \neq c]
    \Rightarrow \\
    & \quad [
    \vardeltaoutlen(\trnsstate, c) = 0 \Rightarrow
    \vareddist(\trnsstate, c) = 1 \wedge \\
    & \quad \vardeltaoutlen(\trnsstate, c) \neq 0 \Rightarrow
    \vareddist(\trnsstate, c) = \vardeltaoutlen(\trnsstate, c) ]
\end{split}
\end{equation}

\mypar{Edit Distance of Arbitrary Strings.}
Suppose that $T$ has the transitions $q_0 \xrightarrow{\texttt{a}/\texttt{a}} q1 \xrightarrow{\texttt{a}/\texttt{bc}} q_2$, and the specified mean edit distance is $\distancebound=0.5$.
The edit distance is 0 for the first transition and 2 for the second one.
For the input string $\texttt{aa}$, the mean aggregate cost is $2 / 2$, which means that the specification is not satisfied.
In general, we cannot keep track of every input string in the input type and look at its length and the number of edits that were made over it.
So, how can we compute the mean aggregate cost over any input string?
The first part of our solution is to scale the edit distance over a single transition depending on the specified mean edit distance.
This operation makes it such that an input string is under the edit distance bound if the sum  of the weighted edit distances of its transitions is $\geq 0$.
The invariant we need to maintain is that the sum of the weights at any stage of the run gives us where we are with regard to the mean aggregate cost. 
For each transition we compute the difference between the edit distance over the transition and the specified mean edit distance $\distancebound$.
We introduce the uninterpreted function $\varweight : \settrnsstate \times \Sigma \rightarrow \mathbb{Q}$, which stands for weighted edit distance.
For a transition at $\trnsstate$ reading a character $c$, the weighted edit distance is given by $\varweight(\trnsstate, c) = \distancebound - \vareddist(\trnsstate, c)$.
The sum of the weights of all transitions tells us the cumulative difference.
Going back to our example, the weighted edit distances of the two transitions are $\varweight(q_0,\texttt{a})=0.5$ and $\varweight(q_1,\texttt{a})=-1.5$, making the cumulative distance $-1$ and implying that the specification is violated.
We can now compute the mean edit distance over a run without keeping track of the length of the run and the number of edits performed over it.

We still need to compute the weighted edit distance for every string in the possibly infinite language $\DFAlanguage{\inputtype}$.
Building on the idea of simulation from the previous section, we introduce a new function called $\varprogress: \setsrcstate \times \settrnsstate \times \settrgstate \rightarrow \mathbb{Q}$, which tracks an upper bound on the sum of the distances so far at that point in the simulation.
This function is similar to a \textit{progress measure}, which is a type of invariant used to solve \textit{energy games} \cite{progressmeasure}, a connection we expand on in Section~\ref{sec:related_work}.
In particular, we already know that if there exist strings $w,w'$ such that
    $\extendedtransition{\inputtype}(\initstate{\inputtype}, w) = q_{\inputtype}$, 
    $\extendeddeltaTstate(\initstate{T}, w) = \trnsstate$, 
    $\extendeddeltaTout(\initstate{T}, w) = w'$, 
    and $\extendedtransition{\outputtype}(\initstate{\outputtype}, w') = q_{\outputtype}$, then we have $\simulation{\srcstate}{\trnsstate}{\trgstate}$.
Let this run over $T$ be denoted by $\initstate{T} \xrightarrow{a_0/y_0} \trnsstate^1 \ldots \trnsstate^{n-1} \xrightarrow{a_{n-1}/y_{n-1}} \trnsstate$, where $w = a_0\cdots a_{n-1}$, $w' = y_0\cdots y_{n-1}$, and $\trnsstate = \trnsstate^n$.
We have that $\varprogress(\srcstate, \trnsstate, \trgstate) \geq \sum_{i = 0}^{n-1} \varweight(\trnsstate^i, a_i)$.

The \varprogress{} function is a budget on the number of  edits we can still perform.
At the initial states, we start with no `initial credit' and the energy is 0.
\begin{equation}
\label{eq:distance4}
    \varprogress(\initstate{\inputtype}, \initstate{T}, \initstate{\outputtype}) = 0
\end{equation}

Constraint~\ref{eq:distance5} bounds the energy budget according to the weighted edit distance of a transition by computing the minimum budget required at any point to still satisfy the distance bound. For each combination of $\srcstate,\trnsstate,\trgstate$, and $c \in \Sigma$, the constraint uses free variables $c_0, \ldots, c_{l}$ and $\trgstate^0, \ldots, \trgstate^{l-1}$:
\begin{equation}
\label{eq:distance5}
\begin{split}
    \hspace{-18mm}&\bigwedge_{\mathclap{0 \leq z < l}}  (\vardeltaoutlen(\trnsstate, c) {=} z \Rightarrow  \\
    & [\bigwedge_{\mathclap{0 \leq x < z}} \vardeltaout(\trnsstate, c, x) {=} c_x] {\wedge} 
    [\trgstate^0 {=} \trgstate \wedge \bigwedge_{\mathclap{1 \leq x < z}} \trgstate^x {=} 
    \vardeltaoutputtype(\trgstate^{x-1},c_{x-1})] \wedge \\
    & \energyfunc{\srcstate}{\trnsstate}{\trgstate} \geq \energyfunc{\transition{\inputtype}(\srcstate, c)}{\vardeltastate(\trnsstate, c)}{\trgstate^z} {-} 
    \varweight(\trnsstate, c))
\end{split}
\end{equation}

In our example, Constraint~\ref{eq:distance5} encodes that the energy at $q_0$ can be 1 less than that at $q_1$, but that the energy at $q_1$ needs to be $3$ greater than at $q_2$ since we need to spend 3 edit operations over the second transition.

At any point during a run, the transducer is allowed to go below the mean edit distance and then `catch up' later because we only care about the edit distance when the transducer has finished reading a string in $\DFAlanguage{\inputtype}$.
Therefore, when we reach a final state of $\inputtype$, the transducer should not be in `energy debt'.
\begin{equation} 
\label{eq:distance6}
    \bigwedge_{\srcstate \in \finalstates{\inputtype}} \simulation{\srcstate}{\trnsstate}{\trgstate} \Rightarrow \energyfunc{\srcstate}{\trnsstate}{\trgstate} \geq 0
\end{equation}

\noindent The encoding presented in this section is sound (Proofs in~\ref{sec:distances_proofs}).


\section{Richer Models and Specifications}
\label{sec:extensions}

We extend our technique to more expressive models (Sections~\ref{sec:symbolic_extension} and~\ref{sec:regular-look-ahead}) and show how our synthesis approach can be used not only to synthesize transducers, but also to repair them (Section~\ref{sec:repair}).
Furthermore, in Appendix~\ref{sec:distances}, we describe an encoding of an alternative distance measure.

\subsection{Symbolic Transducers} 
\label{sec:symbolic_extension}

Symbolic finite automata (s-FA) and transducers  (s-FT) extend their non-sym\-bo\-lic counterparts by allowing transitions to carry predicates and functions to represent (potentially infinite) sets of input characters and output strings.
Figure~\ref{fig:symbolic_transducer} shows an s-FT that extends the escapeQuotes transducer from Figure~\ref{fig:escape_finite_correct} to handle alphabetic characters.
The bottom transition from $q_0$ reads a character \texttt{"} (bound to the variable $x$) and outputs the string \texttt{\textbackslash "} (i.e., a \texttt{\textbackslash} followed by the character stored in $x$).
Symbolic finite automata (s-FA) are s-FTs with no outputs.
To simplify our exposition, we focus on s-FAs and s-FTs that only operate over ASCII characters that are ordered by their codes.
In particular, all of our predicates are unions of intervals over characters (i.e., $ x \neq \texttt{\textbackslash}$ is really the union of intervals [\texttt{NUL}-\texttt{[}] and [\texttt{]}-\texttt{DEL}]); we often use the predicate notation instead of explicitly writing the intervals for ease of presentation.
Furthermore, we only consider two types of output functions: constant characters and offset functions of the form $x+k$ that output the character obtained by taking the input $x$ and adding a constant $k$ to it---e.g., applying $x+(-32)$ to a lowercase alphabetic letter gives the corresponding uppercase letter.

In the rest of the section, we show how we can solve the transducer synthesis problem in the case where $\inputtype$ and $\outputtype$ are s-FAs and the goal is to synthesize an s-FT (instead of an FT) that meets the given specification.
Intuitively, we do this by `finitizing' the alphabet of the now symbolic input-output types, synthesizing a finite transducer over this alphabet using the technique presented in Section~\ref{sec:encoding}, and then extracting an s-FT from the solution.

\mypar{Finitizing the Alphabet.}
The idea of finitizing the alphabet of s-FAs is a known one~\cite{cacmSFA} and is based on the concept of $\minterms$, which is the set of maximal satisfiable Boolean combinations of the predicates appearing in the s-FAs.
For an s-FA $\sfa$, we can define its set of predicates as: $\mathit{Predicates}(\sfa) = \{\phi \mid q \xrightarrow{\phi} q' \in \transition{\sfa} \}$.
The set of minterms $\mintermsshort(\sfa)$ is the set of satisfiable Boolean combinations of all the predicates in $\mathit{Predicates}(\sfa)$.
For example, for the set of predicates over the s-FT escapeQuotes in Figure~\ref{fig:symbolic_transducer}, we have that $\mintermsshort(\text{escapeQuotes}) =\{x \neq \texttt{"} \wedge x \neq \texttt{\textbackslash}, 
            x = \texttt{"}, 
            x = \texttt{\textbackslash}\}$.
The reader can learn more about minterms in \cite{cacmSFA}.
We assign each minterm a representative character, as indicated in Figure~\ref{fig:minterms}, and then construct a finite automaton from the resulting finite alphabet $\Sigma$.
For a character $c \in \Sigma$, we refer to its corresponding minterm by $\mintermpred{c}$.
In the other direction, for each minterm $\psi \in \minterms(\sfa)$, we refer to its uniquely determined representative character by $\witnessshort{\psi}$.

\begin{figure}[t]
  \captionsetup[subfigure]{justification=centering}
  \centering
  \begin{subfigure}{0.23\textwidth}
  \centering
  \resizebox{0.98\textwidth}{0.45\textwidth}{
  \begin{tikzpicture}[shorten >=1pt, node distance=3cm, auto] 
    \node[state, accepting, initial] (q0) {$q_0$};
    \node[state, accepting, right of=q0] (q1) {$q_1$};
    
    \path[->]  
            (q0) edge[loop above] node{$ x \neq \texttt{"} \wedge x \neq \texttt{\textbackslash}$ $\rightarrow$ $x$} (q0)
            
            (q0) edge[loop below] node{$ x = \texttt{"}$ $\rightarrow$ \texttt{\textbackslash}$x$} (q0)
            
            (q0) edge node[above] {$ x = \texttt{\textbackslash}$ $\rightarrow$ $x$} (q1)
            
            (q1) edge[bend left] node{$ x \neq \texttt{\textbackslash}$ $\rightarrow$ $x$} (q0)
            
            (q1) edge[bend right, above] node {$ x = \texttt{\textbackslash}$ $\rightarrow$ $x$} (q0)
            ;
    \end{tikzpicture}
  }
  \caption{escapeQuotes s-FT}
  \label{fig:symbolic_transducer}
  \end{subfigure}
  ~
  \begin{subfigure}{0.23\textwidth}
  \centering
  \resizebox{0.98\textwidth}{0.45\textwidth}{  
  \begin{tikzpicture}[shorten >=1pt, node distance=2.5cm, auto] 
    \node[state, accepting, initial] (q0) {$q_0$};
    \node[state, accepting, right of=q0] (q1) {$q_1$};
    \path[->]  
            (q0) edge[loop above] node{\texttt{a} $\rightarrow$ \texttt{a}} (q0)
            
            (q0) edge[loop below] node{\texttt{"} $\rightarrow$ \texttt{\textbackslash "}} (q0)
            
            (q0) edge node[above] {\texttt{\textbackslash} $\rightarrow$ \texttt{\textbackslash}} (q1)
            
            (q1) edge[bend left=55] node{\texttt{a} $\rightarrow$ \texttt{a}} (q0)
            
            (q1) edge[bend left=20] node{\texttt{"} $\rightarrow$ \texttt{"}} (q0)
            
            (q1) edge[bend right=45, above] node {\texttt{\textbackslash} $\rightarrow$ \texttt{\textbackslash}} (q0)
            ;
    \end{tikzpicture}
    }
    \caption{$\finite{escapeQuotes}$}
    \label{fig:escape_finite}
    \end{subfigure}
    
    \vspace{1ex}
    \begin{subfigure}{0.47\textwidth}
    \centering
    \fbox{%
    \scriptsize
        \parbox{0.95\textwidth}{%
            \textbf{minterms:}
            $ [x \neq \texttt{"} \wedge x \neq \texttt{\textbackslash}] $, 
            $ [x = \texttt{"}] $, 
            $ [x = \texttt{\textbackslash}] $ \\
            \textbf{witness char:}
            $ \witnessshort{[x \neq \texttt{"} \wedge x \neq \texttt{\textbackslash}]} {=} \texttt{a} $,
            $ \witnessshort{[x = \texttt{"}]} {=} \texttt{"} $,
            $ \witnessshort{[x = \texttt{\textbackslash}]} {=} \texttt{\textbackslash} $
        }%
    }
    \caption{Set of minterms and their witness elements}
    \label{fig:minterms}
    \end{subfigure}
\caption{Example of Finitization}
\label{fig:minterm_reduction}
\end{figure}

For an s-FA $\sfa$, we denote its corresponding FA over the alphabet $\mintermsshort(\sfa)$ with $\finite{M}$.
Given an s-FA $\sfa$, the set of transitions of $\finite{\sfa}$ is defined as follows:
\[
\transition{\finite{\sfa}} {=}  \{q \xrightarrow{\witnessshort{\psi}} q' {\mid} q \xrightarrow{\phi} q' \wedge \psi \in \mintermsshort(\sfa) \wedge \text{IsSat}
        (\psi \land \phi)\}
\]
This algorithm replaces a transition guarded by a predicate $\phi$ in the given s-FA with a set of transitions consisting of the witnesses of the minterms where $\phi$ is satisfiable.
In interval arithmetic this is the set of intervals that intersect with the interval specified by $\phi$.
The transition from $q_1$ guarded by the predicate $[x \neq \texttt{\textbackslash}]$ in Figure~\ref{fig:symbolic_transducer} intersects with 2 minterms  $[x \neq \texttt{"} \wedge x \neq \texttt{\textbackslash}]$ and $[x = \texttt{"}]$.
As a result, we see that this transition is replaced by two transitions in Figure~\ref{fig:escape_finite}, one that reads \texttt{"} and another that reads \texttt{a}.

\mypar{From FTs to s-FTs.}
Once we have synthesized an FT $T$, we need to extract an s-FT from it.
There are many s-FTs equivalent to a given FT and here we present one way of doing this conversion which is used in our implementation.
Let the size of an interval $I$ (the number of characters it contains) be given by $\intervalsize(I)$, and the offset between 2 intervals $I_1$ and $I_2$ (i.e. the difference between the least elements of $I_1$ and $I_2$) be given by $\mathit{offset}(I_1,I_2)$. 
Suppose we have a transition $q \xrightarrow{c/y_0\cdots y_n} q'$, where $c, y_i \in \Sigma$. 
Then, we construct a transition $q \xrightarrow{\mintermpred{c}/f_0\cdots f_n} q'$, where for each $y_i$, the corresponding function $f_i$ is determined by the following rules ($x$ always indicates variable bound to the input predicate): 
\begin{enumerate}[topsep=1pt]
    \item If $c = y_i$, then $f_i = (x)$, i.e. the identity function.

    \item If $\mintermpred{c}$ and $\mintermpred{y_i}$ consist of single intervals $I_1$ and $I_2$, respectively, such that $\intervalsize(I_1) = \intervalsize(I_2)$ , then $f_i = (x + \mathit{offset}(I_1,I_2))$.
    For instance, if the input interval is \texttt{[a-z]} and the output interval is \texttt{[A-Z]}, then the output function is $(x + (-32))$, which maps lowercase letters to uppercase ones.
    
    \item Otherwise $f_i = y_i$---i.e., the output is a  character in the output minterm.
\end{enumerate}

While our s-FT recovery algorithm is sound, it may apply case 3 more often than necessary and introduce many constants, therefore yielding a transducer that does not generalize well to unseen examples.
However, our evaluation shows that our technique works well in practice.
We provide a sketch of the proof of soundness of this algorithm in Appendix~\ref{sec:sft_synthesis_soundness}.

\subsection{Synthesizing Transducers with Lookahead}
\label{sec:regular-look-ahead}

Deterministic transducers cannot express functions where the output at a certain transition depends on future characters in the input.
Consider the problem of extracting all substrings of the form \texttt{<x>} (where $\texttt{x} \neq \texttt{<}$) from an input string. 
This is the \emph{getTags} problem from \cite{POPL12SFT}.
A deterministic transducer cannot express this transformation because when it reads \texttt{<} followed by \texttt{x} it has to output \texttt{<x} if the next character is a \texttt{>} and nothing otherwise.
However, the transducer does not have access to the next character! 

%
Instead, we extend our technique to handle deterministic transducers with lookahead, i.e., the ability to look at the string suffix when reading a symbol.
Formally, a \textit{Transducer with Regular Lookahead} is a pair $(T, \lookaheadaut)$ where $T$ is an FT with $\Sigma_T = \setlookaheadstate \times \Sigma$, and $\lookaheadaut$ is a total DFA with $\Sigma_{\lookaheadaut} = \Sigma$.
The transducer $T$ now has another input in its transition function, although it still only outputs characters from $\Sigma$, i.e., $\deltaTout:\settrnsstate \times (\setlookaheadstate \times \Sigma) \rightarrow \Sigma$, and $\deltaTstate:\settrnsstate \times (\setlookaheadstate \times \Sigma) \rightarrow \settrnsstate$.
The semantics is defined as follows.
Given a string $w = a_0\cdots a_{n}$, we define a function $\lookahead_w$ such that $\lookahead_w(i)=\transition{\lookaheadaut} (\initstate{\lookaheadaut}, a_{n}\cdots a_{i+1})$.
In other words, $\lookahead_w(i)$ gives the state reached by $\lookaheadaut$ on the reversed suffix starting at $i + 1$.
%
At each step $i$, the transducer $T$ reads the symbol $(a_i,\lookahead_w(i))$.
The extended transition functions now take as input a lookahead word, which is a sequence of pairs of lookahead states and characters, i.e., from $(\setlookaheadstate \times \Sigma)^\ast$.

To synthesize transducers with lookahead, we introduce uninterpreted functions $\vardeltalookahead$ for the transition function of $\lookaheadaut$, and $\varlookahead_{w}$ for the $r$-values of $w$ on $\lookaheadaut$.
Additionally, we introduce a bound $\numlookaheadstates$ on the number of states in the lookahead automaton $R$ as our synthesis algorithm has to synthesize both $T$ and $R$ at the same time.
Appendix~\ref{sec:lookahead_constraints} provides the modified constraints needed to encode input-output types and input-output examples to use lookahead.

Part of the transducer with lookahead we synthesize for the getTags problem is shown in Figure~\ref{fig:get_tags_lookahead}.
Notice that there are 2 transitions out of $q_1$ for the same input but different lookahead state: the string \texttt{<x} is outputted when the lookahead state is $r_1$.

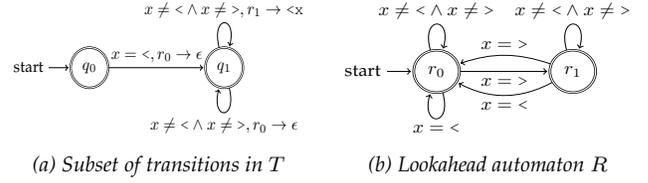
\begin{figure}[t]
  \captionsetup[subfigure]{justification=centering}
  \centering
  \begin{subfigure}{0.23\textwidth}
  \centering
  \resizebox{0.98\textwidth}{0.45\textwidth}{
  \begin{tikzpicture}[shorten >=1pt, node distance=3cm, auto]
    \node[state, accepting, initial] (q0) {$q_0$};
    \node[state, accepting, right of=q0] (q1) {$q_1$};
    
    \path[->] 
    (q0) edge node {$x = \texttt{<}, r_0 \rightarrow \epsilon$} (q1)
    
    (q1) edge[loop above] node {$ x \neq \texttt{<} \wedge x \neq \texttt{>}, r_1 \rightarrow \texttt{<x}$} (q1)
    
    (q1) edge[loop below] node {$ x \neq \texttt{<} \wedge x \neq \texttt{>}, r_0 \rightarrow \epsilon$} (q1)
    ;
  \end{tikzpicture}
  }
  \caption{Subset of transitions in $T$}
  \label{fig:transducer_lookahead}
  \end{subfigure}~
  \begin{subfigure}{0.23\textwidth}
  \centering
  \resizebox{0.98\textwidth}{0.45\textwidth}{
  \begin{tikzpicture}[shorten >=1pt, node distance=2.5cm, auto]
    \node[state, accepting, initial] (r0) {$r_0$};
    \node[state, accepting, right of=q0] (r1) {$r_1$};
    
    \path[->]  
            (r0) edge[loop above] node{$ x \neq \texttt{<} \wedge x \neq \texttt{>}$} (r0)
            
            (r0) edge[loop below] node{$ x = \texttt{<}$} (r0)
            
            (r0) edge node[below] {$ x = \texttt{>}$} (r1)
            
            (r1) edge[loop above] node {$ x \neq \texttt{<} \wedge x \neq \texttt{>}$} (r1)
            
            (r1) edge[bend right=20] node[above] {$ x = \texttt{>}$} (r0)
            
            (r1) edge[bend left=30] node[below] {$ x = \texttt{<}$} (r0)
            ;
    \end{tikzpicture}
    }
    \caption{Lookahead automaton $\lookaheadaut$}
    \label{fig:lookahead_aut}
    \end{subfigure}
  
  \caption{Regular lookahead for getTags}
  \label{fig:get_tags_lookahead}
\end{figure}

\mypar{Lookahead and aggregate cost:} 
Lookahead can help representing transducers, even deterministic ones, in a way that has lower aggregate cost---i.e., the aggregate cost better approximates the actual edit distance.
Suppose that we want to synthesize a transducer that translates the string \texttt{abc} to \texttt{ab} and the string \texttt{abd} to \texttt{bd}.
This translation can be done using a deterministic transducer with transitions $q_0 \xrightarrow{a/\epsilon} q_1 \xrightarrow{b/\epsilon} q_2$, followed by two transitions from $q_2$ that choose the correct output based on the next character.
Such a transducer would have a high aggregate cost of 4, even though the actual edit distance is 1.
In contrast, using lookahead we can obtain a transducer that can output each character when reading it; this transducer will have aggregate cost 1 for either string.
We conjecture that for every transducer $T$, there always exists an equivalent transducer with regular lookahead $(T',R)$ for which the edit distance computation for aggregate cost coincides with the actual edit distance of $T$.
%

\subsection{Transducer Repair}
\label{sec:repair}

In this section, we show how our synthesis technique can also be used to ``repair'' buggy transducers.
The key idea is to use the closure properties of automata and transducers---e.g., closure under union and sequential compositions~\cite{cacmSFA}---to reduce repair problems to synthesis ones.
The ability to algebraically manipulate transducers and automata is one of the key aspects that distinguishes our work from other synthesis works that use domain-specific languages~\cite{gulwani2011flashfill,symmetriclenses}.

We describe two settings in which we can repair an incorrect transducer $\bad$:
\noindent \textbf{1.} Let $\hoaretriple{\inputtype}{\bad}{\outputtype}$ be an input-output type violated by $\bad$ and let $\outputlang{\inputtype}{\bad}$ be the finite automaton describing the set of strings $\bad$ can output when fed inputs in $\inputtype$ (this is computable thanks to closure properties of transducers).
We are interested in the case where $\outputlang{\inputtype}{\bad}\setminus \outputtype\neq \emptyset$---i.e., $\bad$ can produce strings that are not in the output type.

\noindent \textbf{2.} Let $\exampleset$ be a set of input-output examples. We are interested in the case where there is some example $\examplemap$ such that $\bad(\inputexample) \neq \outputexample$.

We describe two ways in which we repair transducers.



\mypar{Repairing from the Input Language.}
This approach synthesizes a new transducer for the inputs on which $\bad$ is incorrect.
Using properties of transducers, we can compute an automaton describing the exact set of inputs $\badinput\subseteq \inputtype$ for which $\bad$ does not produce an output in $\outputtype$ (see pre-image computation in~\cite{LorisCAV17}).
Let $restrict(T,L)$ be the transducer that behaves as $T$ if the input is in $L$ and does not produce an output otherwise (transducers are closed under restriction~\cite{LorisCAV17}).
If we synthesize a transducer $T_1$ with type $\hoaretriple{\badinput}{T_1}{\outputtype}$, then the transducer $restrict(T_1,\badinput)\cup restrict(\bad,\inputtype\setminus\badinput)$ satisfies the desired input-output type (transducers are closed under union).

\mypar{Fault Localization from Examples.}
We use this technique when $\bad$ is incorrect on some example.
We can compute a set of ``suspicious'' transitions by taking all the transitions traversed when $T(\inputexample) \neq \outputexample$ for some $\examplemap\in E$ (i.e., one of these transitions is wrong) and removing all the transitions traversed when $T(\inputexample) = \outputexample$ for some $\examplemap\in E$ (i.e., transitions that are likely correct).
Essentially, this is a way of identifying $\badinput$ when $\bad$ is wrong on some examples.
We can also use this technique to limit the transitions we need to synthesize when performing repair.

\section{Evaluation}
\label{sec:eval}

We implemented our technique in a Java tool \name (Automatic Synthesis of TRAnsducers), which uses Z3 \cite{z3_2008} to solve the generated constraints.
We evaluate using a 2.7 GHz Intel Core i5, RAM 8 GB, with a 300s timeout.



\subsection*{Q1: Can \name synthesize practical transformations?}
\label{sec:eval-synthesis}

\textit{Benchmarks.}
Our first set of benchmarks is obtained from Optician~\cite{symmetriclenses,bijectivelenses}, a tool for synthesizing lenses, which are bidirectional programs used for keeping files in different data formats synchronized.
We adapted 11 of these benchmarks to work with \name (note that we only synthesize one-directional transformations), and added one additional benchmark extrAcronym2, which is a harder variation (with a larger input type) of extrAcronym.
We excluded benchmarks that require some memory, e.g.,  swapping words in a sentence, as they cannot be modeled with transducers.

Our second set of benchmarks (Miscellaneous) consists of 6 problems we created based on file transformation tasks (unixToDos, dosToUnix and CSVSeparator), and s-FTs from the literature--escapeQuotes from \cite{bekonline}, getTags and quicktimeMerger from \cite{POPL12SFT}.
All of the benchmarks require synthesizing s-FTs and getTags requires synthesizing an s-FT with lookahead (details in Table~\ref{table:synthbenchmarks}).

To generate the sets of examples, we started with the examples that were used in the original source when available.
In 5 cases, \name synthesized a transducer that was not equivalent to the one synthesized by Optician, even though it did meet the specification.
In these cases, we used \name to synthesize two different transducers that met the specification, computed a string on which the two transducers differed, and added the desired output for that string as a new example.
We repeated this task until \name yielded the desired transducer and we report the time for such sets of examples.
The number of examples varies between 1 and 5.
The ability to perform equivalence checks of two transducers is  yet another reason why synthesizing transducers is useful and is in some ways preferable to synthesizing programs in a DSL.
For each benchmark we chose a mean edit distance of 0.5 when the transformation could be synthesized with this distance and of 1 otherwise.

\begin{table*}[!htbp]\centering
\caption{Metrics of \name's performance on the set of synthesis benchmarks. The right-most set of columns gives the synthesis time for \name and Optician (under 2 different configurations).
The middle set of columns gives the sizes of the parameters to the synthesis problem. 
In particular, $\setsrcstate$ and $\settrgstate$ denote the number of input and output states and $\transition{\inputtype}$ and $\transition{\outputtype}$ denote the number of transitions in the input and output types, respectively.
A \xmark represents a benchmark that failed.
\NA\ stands in for data that is not available; this is because we only re-ran Optician on the benchmarks that were already encoded in its benchmark set, plus a few additional ones for comparing between the tools that we wrote ourselves.}
\label{table:synthbenchmarks}
\small
\begin{tabular}{c|l|rrrrrrrrr|rrr}\toprule
\multirow{2}{*}[-0.4ex]{} & & & & & & & & & & & & \quad Time(s) & \\
& Benchmark & $\setsrcstate$ & $\settrgstate$ &
$\transition{\inputtype}$ & $\transition{\outputtype}$ & $\Sigma$ & $E$ &$k$ &$l$ & $\distancebound$ &\name &Optician & Optician-re \\\midrule
\multirow{12}{*}[-0.4ex]{\rotatebox{90}{Optician}} & extrAcronym &6 &3 &10 &3 &3 &2 &1 &1 & .5 &0.11 &0.05 & \xmark \\
& extrAcronym2 &6 &3 &16 &3 &3 &3 &2 &1 & 1 &0.42 & \NA & \NA \\
& extrNum &15 &13 &17 &12 &3 &1 &1 &1 & 1 &0.93 &0.05 &0.07 \\
& extrQuant &4 &3 &8 &5 &2 &1 &2 &1 & 1 &0.19 &0.09 & \xmark \\
& normalizeSpaces &7 &6 &19 &10 &2 &2 &2 &1 & 1 &0.46 &16.64 &\xmark \\
& extrOdds &15 &9 &29 &13 &5 &3 &3 &2 & 1 &15.87 &0.12 &\xmark \\
& capProb &3 &3 &3 &3 &2 &2 &2 &1 & 1 &0.05 &0.05 &\xmark \\
& removeLast &6 &3 &8 &3 &3 &3 &2 &1 & .5 &0.21 &0.15 &0.07 \\
& sourceToViews &18 &7 &26 &15 &5 &3 &3 &2 & 1 &50.92 &0.06 &\xmark \\
& normalizeNamePos &19 &7 &35 &24 &13 & 1 & 6 & 2 & 1 & \xmark &0.05 &0.10 \\
& titleConverter &22 &13 &41 &41 &15 & 1 & 3 & 1 & 1 & \xmark &0.07 &\xmark \\
& bibtextToReadable &14 &11 &41 &35 &12 & 1 & 5 & 1 & 1 & \xmark &0.64 &0.15 \\
\hline
\multirow{6}{*}[-0.4ex]{\rotatebox{90}{Miscellaneous}} & unixToDos &5 &7 &17 &19 &4 &4 &2 &2 & .5 &1.24 & \NA & \NA \\
& dosToUnix &7 &5 &19 &17 &4 &4 &2 &1 & .5 &0.41 & \NA & \NA \\
& CSVSeparator &5 &5 &9 &9 &4 &1 &1 &1 & 1 &0.142 & \NA & \NA \\
& escapeQuotes &2 &2 &6 &5 &3 &5 &2 &2 & 1 &0.188 & \xmark & \xmark \\
& quicktimeMerger &7 &3 &9 &3 &2 &2 &1 &1 & .5 &0.075 & \NA & \NA \\
& getTags & 3& 3& 9& 4& 3& 5& 2& 2& 1& 0.95& \xmark & \xmark \\
\bottomrule
\end{tabular}
\end{table*}

\mypar{Effectiveness of \name.}
\name can solve $15/\numsynthbenchmarks$ benchmarks (13 in <1s and 2 under a minute) and times out on 3 benchmarks where both $P$ and $Q$ are big.

While the synthesized transducers have at most 3 states, we note that this is because \name synthesizes total transducers and then restricts their domains to the input type $\inputtype$.
This is advantageous because synthesizing small total transducers is easier than synthesizing transducers that require more states to define the domain.
For instance, when we restrict the solution of extrAcronym2 to its input type, the resulting transducer has 11 states instead of the 2 required by the original solution!

\mypar{Comparison with Optician.}
We do not  compare \name to tools that only support input-output examples.
Instead, we compare \name to Optician on the set of benchmarks common to both tools.
Like \name, Optician supports input-output examples and types, but the types are expressed as regular expressions.
Furthermore, Optician also attempts to produce a program that minimizes a fixed information theoretical distance between the input and output types~\cite{symmetriclenses}.

Optician is faster when the number of variables in the constraint encoding increases, while \name is faster on the normalizeSpaces benchmark.
Optician, which uses regular expressions to express the input and output types, does not work so well with unstructured data.
To confirm this trend, we wrote synthesis tasks for the escapeQuotes and getTags benchmarks in Optician and it was unable to synthesize those as well---e.g., escapeQuotes requires replacing every \texttt{"} character with \texttt{\textbackslash "}.

To further look at the reliance of Optician on regular expressions, we converted the regular expressions used in the lens synthesis benchmarks to automata and then back to regular expressions using a variant of the state elimination algorithm that acts on character intervals.
This results in regular expressions that are not very concise and might have redundancies.
Optician could only solve 4/11 benchmarks that it was previously synthesizing (Optician-re in Table~\ref{table:synthbenchmarks}).

\textbf{Answer to Q1:} \name can solve real-world benchmarks and has performance comparable to that of Optician for similar tasks. Unlike Optician, \name does not suffer from variations in how the input and output types are specified.

\subsection*{Q2: Can \name repair transducers in practice?}
\label{sec:eval-repair}
\textit{Benchmarks.}
We considered the benchmarks in Table~\ref{table:repairbenchmarks}.
The only pre-existing benchmark that we found was escapeQuotes, through the interface of the Bek programming language used for verifying transducers \cite{bekonline}.
We generated 11 additional faulty transducers to repair in the following two ways:
\rone Introducing faults in our synthesis benchmarks: We either replaced the output string of a transition with a constant character, inserted an extra character, or deleted a transition altogether.
\rtwo Incorrect transducers: We intentionally provided fewer input-output examples and used only example-based constraints on some of our synthesis benchmarks.

All the benchmarks involve s-FTs.
Three benchmarks are wrong on both input-output types and examples and the rest are only wrong on examples.
Additionally, we note that to repair a transducer, we need the ``right'' set of minterms. 
Typically, the set of minterms extracted from the transducer predicates is the right one, but in the case of the escapeBrackets problems, \name needs a set of custom minterms we provide manually---i.e., repairing the transducer requires coming up with a new predicate.
%
We are not aware of another tool that solves transducer repair problems and so do not show any comparisons.

\begin{table*}[htbp!]\centering
\caption{Metrics of \name's performance on the set of repair benchmarks. The two right-most columns give the synthesis time without and with the use of templates.
Default is the case where a new transducer is synthesized for $\badinput$ and Template is the case where a partial solution to the solver is provided.
The $\transition{\bad}$ column gives the number of transitions that were localized by the fault-localization procedure as a fraction of the total number of transitions in the transducer. 
The other columns that describe the parameters of the synthesis problem in the default case are the same as for Table~\ref{table:synthbenchmarks}. }
\label{table:repairbenchmarks}
\small
\begin{tabular}{c|l|rrrrrrrrrr|rr}\toprule
\multirow{2}{*}[-0.4ex]{} & & & & & & & & & & & & {Time(s)} & \\
& Benchmark & $\setsrcstate$ & $\settrgstate$ &
$\transition{\inputtype}$ & $\transition{\outputtype}$ & $\Sigma$ & $E$ &$k$ &$l$ & $\distancebound$ & $\transition{\bad}$ &Default &Template  \\\midrule
\multirow{7}{*}[-0.4ex]{\rotatebox{90}{Fault injected}} & swapCase1 &2 &1 &6 &3 &3 &2 &1 &1 & 1 &3/3 &0.04 &0.02 \\
& swapCase2 &2 &1 &4 &3 &3 &2 &1 &1 &1 & 1/2 & \xmark & \xmark \\
& swapCase3 &2 &1 &6 &3 &3 &2 &1 &1 & 1 &1/3 &0.06 &0.05 \\
& escapeBrackets1 &2 &6 &16 &36 &8 &4 &1 &4 & 4 &1/3 &0.69 &0.42 \\
& escapeBrackets2 & 1& 6& 1& 7& 6& 5& 1& 4& 4 & 1/2 & \xmark & \xmark \\
& escapeBrackets3 &2 &7 &8 &36 &9 &5 &1 &4 & 4 &2/3 &1.12 &0.34 \\
& caesarCipher &2 &1 &4 &2 &3 &1 &1 &1 &1 &1/1 & \xmark & \xmark \\
\hline
\multirow{4}{*}[-0.4ex]{\rotatebox{90}{Synth.}} & extrAcronym2 &11 &3 &30 &3 &3 &3 &2 &1 & 1 &12/30 &0.59 &10.15 \\
& capProb &3 &3 &3 &3 &2 &2 &2 &1 & 1 &3/3 &0.04 &0.04 \\
& extrQuant &8 &3 &16 &5 &2 &1 &2 &1 & 1 &5/10 &0.37 &0.51 \\
& removeLast &6 &3 &8 &3 &3 &2 &2 &1 & .5 &7/8 &0.40 &1.08 \\
\hline
& escapeQuotes & 3& 2& 9& 5& 3 & 5 & 2 & 1 & 1 & 3/5& 0.17 & 0.10\\
\bottomrule
\end{tabular}
\end{table*}

\mypar{Effectiveness of \name.}
We indicate the number of suspicious transitions identified by our fault localization procedure (Section~\ref{sec:repair}) in the column labeled $\transition{\bad}$.
In many cases, \name can detect 50\% of the transitions or more as being likely correct, therefore reducing the space of unknowns.

We compare 2 different ways of solving repair problems in \name.
One uses the repair-from-input approach described in Section~\ref{sec:repair} (Default in Table~\ref{table:repairbenchmarks}).
The second approach involves using a `template', where we supply the constraint solver with a partial solution to the synthesis problem, based on the transitions that were localized as potentially buggy (Template in Table~\ref{table:repairbenchmarks}).

\name can solve 9/$\numrepairbenchmarks$ repair benchmarks (all in less than 1 second).
The times using either approach are comparable in most cases.
While one might expect templates to be faster, this is not always the case because the input-output specification for the repair transducer is small, but providing a template requires actually providing a partial solution, which in some cases happens to involve many constraints.

\textbf{Answer to Q2:} \name can repair transducers with varying types of bugs.

\section{Related Work}
\label{sec:related_work}
\textit{Synthesis of string transformations.}
String transformations are one of the main targets of program synthesis.
Gulwani showed they could be synthesized from input-output examples \cite{gulwani2011flashfill} and introduced the idea of using a DSL to aid synthesis.
Optician extended the DSL-based idea to synthesizing lenses \cite{bijectivelenses,symmetriclenses}, which are programs that transform between two formats.
Optician supports not only examples but also input-output types.
While DSL-based approaches provide good performance, they are also monolithic as they rely on the structure of the DSL to search efficiently.
\name does not rely on a DSL and can synthesize string transformations from complex specifications that cannot be handled by DSL-based tools.
Moreover, transducers allow applying verification techniques to the synthesized programs (e.g., checking whether two solutions are equivalent).
One limitation of transducers is that they do not have `memory', and consequently \name cannot be used for data-transformation tasks where this is required---e.g., mapping the string \texttt{Firstname Lastname} to \texttt{Lastname, Firstname}---something  Optician can do.
We remark that there exist transducer models with such capabilities~\cite{streaming-string-transducers} and our work lays the foundations to handle complex models in the future.

\vspace{1mm}\noindent\textit{Synthesis of transducers.}
Benedikt et al. studied the `bounded repair problem', where the goal is to determine whether there exists a transducer that maps strings from an input to an output type using a bounded number of edits~\cite{Benedikt11}. 
Their work was the first to identify the relation between solving such a problem and solving games, an idea we leverage in this paper.
However, their work is not implemented, cannot handle input-output examples, and therefore shies away from the source of NP-Completeness.
Hamza et al. studied the problem of synthesizing minimal non-deterministic Mealy machines (transducers where every transition outputs exactly one character), from examples \cite{minimaltransducersynthesis}.
They prove that the problem of synthesizing such transducers is NP-complete and provide an algorithm for computing minimal Mealy machines that are consistent with the input-output examples.
\name is a more general framework that incorporates new specification mechanisms, e.g., input-output types and distances, and uses them all together.
Mealy machines are also synthesized from temporal specifications in reactive synthesis and regular model checking, where they are used to represent parameterized systems~\cite{safety-parametrized-synthesis,lin2016liveness}. 
This setting is orthogonal to ours as the specification is different and the transducer is again only a Mealy machine.

The constraint encoding used in \name is inspired by the encoding presented by Daniel Neider for computing minimal separating DFA, i.e. a DFA that separates two disjoint regular languages \cite{neider2012}.
\name's use of weights and energy to specify a mean edit distance is based on energy games \cite{DBLP:conf/emsoft/ChakrabartiAHS03}, a kind of 2-player infinite game that captures the need for a player to not exceed some available resource.
One way of solving such games is by defining a \emph{progress measure} \cite{progressmeasure}. 
To determine whether a game has a winning strategy for one of the players, it can be checked whether such a progress measure exists in the game.
We showed that the search for such a progress measure can be encoded as an SMT problem.

\section{Conclusion}

We presented a technique and a tool (\name) for synthesizing different types of transducers from types, examples, and distance functions, and showed \name's ability to solve a variety of practical problems.
\name uses SMT solvers and its performance is affected by input components that result in large constraints (e.g., states in the desired transducer).
Because \name synthesizes transducers instead of programs in arbitrary DSLs, its output can be analyzed using transducer algorithms (e.g., equivalence and pre-post analysis).
%
%
%
Because of this property, our approach could be beneficial in learning invariants of string-manipulating programs, where a transducer is the formalism of choice, e.g., in the Ostrich tool \cite{ostrichATVA2020}.
%
%
In terms of improvements to our technique, aside from optimizing the SMT encoding to improve scalability, our approach could be strengthened by devising ways to effectively `guess' the number of states required for a transducer to work on the given inputs. 
We leave these directions for future work.

\section*{Acknowledgements}

This work was funded by the National Science Foundation under grants 1763871, 1750965, 1918211, and 2023222, Facebook and a Microsoft Research Faculty Fellowship.

\bibliographystyle{IEEEtran}
\bibliography{transducer.bib}

\appendices
\label{sec:appendix}

\section{Proofs}
\label{sec:proofs}

\subsection{Input-Output Examples}
\label{sec:examples_proofs}

Lemmas~\ref{lem:example_invariant} and \ref{lem:examples_correctness} show that \rone the $\position_\inputexample$ function preserves a desired invariant, which is used to show that \rtwo for the transducer $T$ encoded in a solution to the constraints, we have $T(\inputexample) = \outputexample$

\begin{lemma}
\label{lem:example_invariant}
Let $\phi$ be the set of example constraints \ref{eq:examples1} and \ref{eq:examples2} for an input-output example $\examplemap$, where $s=a_0\cdots a_n$ and  $t=b_0\cdots b_m$. Let $T$ be a transducer encoded into the functions $\vardeltastate$, $\vardeltaout$, and $\vardeltaoutlen$ as described in Section~\ref{sec:representingTransducers}. If all constraints in $\phi$ are satisfied by $\vardeltastate$, $\vardeltaout$,  $\vardeltaoutlen$, and $\position_\inputexample$, then for all $0 \leq i < n$, $0 \leq j < m$, $\position_\inputexample(i) = (j, \trnsstate)$ iff
$\extendeddeltaTout(\initstate{T}, a_0\cdots a_{i-1}) = b_0\cdots b_{j-1}$ and
 $\extendeddeltaTstate(\initstate{T}, a_0\cdots a_{i-1}) = \trnsstate $.
\end{lemma}

\begin{proof}
We first show the forward direction.

Assume that all constraints in $\phi$ are satisfied by $\vardeltastate$, $\vardeltaout$,  $\vardeltaoutlen$, and $\position_\inputexample$.
We proceed by induction on $i$.
For the base case, when $i = 0$, we have that $\position_\inputexample(0) = (0, \initstate{T})$ by Constraint~\ref{eq:examples1}.
We also know that $\extendeddeltaTout(\initstate{T}, \epsilon) = \epsilon$ and that $\extendeddeltaTstate(\initstate{T}, \epsilon) = \initstate{T}$ by definition of the extended transition function.
So we are done with the base case.

For the induction step, our induction hypothesis states exactly that $\position_\inputexample(i) = (j, \trnsstate)$ iff
$\extendeddeltaTout(\initstate{T}, a_0\cdots \allowbreak{} a_{i-1}) = b_0\cdots b_{j-1}$ and
  $\extendeddeltaTstate(\initstate{T}, a_0\cdots \allowbreak{} a_{i-1}) = \trnsstate $.
Assume we have that $\position_\inputexample(i) = (j, \trnsstate)$.
We need to show that this is the case for $\position_\inputexample(i+1) = (j',\trnsstateTwo)$ as well.
Now, since input words consist of letters in $\Sigma$, we must have that $a_i = c$ for some $c \in \Sigma$. 
Then, we must also have that $\position_\inputexample(i + 1) = (j + \vardeltaoutlen(\trnsstate,c)$, $\vardeltastate(\trnsstate, c))$ by the implication in Constraint~\ref{eq:examples2} and the inductive hypothesis.
By Constraint 2, we have that for all $0 \leq z < l$, either $(\vardeltaout(\trnsstate, c,z) = b_{j+z}$ or $z\geq  \vardeltaoutlen(\trnsstate,c)$. This means that for all $z <  \vardeltaoutlen(\trnsstate,c)$, we have $\vardeltaout(\trnsstate, c,z) = b_{j+z}$. Together with the inductive hypothesis, this implies $\extendeddeltaTout(\initstate{T}, a_0\cdots a_{i-1}) = b_0\cdots b_{j+\vardeltaoutlen(\trnsstate1,a_i)-1}$.

For showing $\extendeddeltaTstate(\initstate{T}, a_0\cdots \allowbreak{}  a_i) = \trnsstateTwo $, we can observe that this follows from the inductive hypothesis that $\extendeddeltaTstate(\initstate{T}, a_0\cdots a_{i-1}) = \trnsstate $, by the definition of $\extendeddeltaTstate$, and that by the constraints of type~\ref{eq:examples2}, which enforce that $\position_\inputexample(i + 1) = (j + \vardeltaoutlen(\trnsstate,c), \vardeltastate(\trnsstate, c))$.
So, we are done with the forward direction.

The backward direction is straightforward. If we have that  $\extendeddeltaTout(\initstate{T}, a_0\cdots \allowbreak{} a_{i-1}) = b_0\cdots b_{j-1}$ and
  $\extendeddeltaTstate(\initstate{T}, a_0\cdots \allowbreak{} a_{i-1}) = \trnsstate $, then there exists a run of the form $\initstate{T} \xrightarrow{a_0 / y_0} \trnsstate^1 
\xrightarrow{a_1 / y_1} \trnsstate^2
\ldots
\trnsstate^{i-1} \xrightarrow{a_{i-1} / y_{i-1}} \trnsstate^{i}$ such that each $y_i$ consists of $\geq 1$ characters and $y_0 \cdots y_{i-1} = b_0\cdots b_{j-1}$. 
We can use each transition to assign corresponding values to $\position_\inputexample(i)$.
\end{proof}

\begin{lemma}
\label{lem:examples_correctness}
Let $T$ be a transducer encoded into $\vardeltastate$, $\vardeltaout$,  $\vardeltaoutlen$,
and $\phi$ be the set of constraints of types \ref{eq:examples1}, \ref{eq:examples2}, and \ref{eq:examples3} for an input-output example $\examplemap$. 
Then $\vardeltastate$, $\vardeltaout$,  $\vardeltaoutlen$ can be extended to a model for  $\phi$ if and only if $T(\inputexample) = \outputexample$.
\end{lemma}

\begin{proof}

We first prove the forward direction. 
Assume that an assignment to $\vardeltastate$, $\vardeltaout$,  $\vardeltaoutlen$, and $\position_\inputexample$ satisfies all constraints in $\phi$.
By Constraint~\ref{eq:examples3}, we have that there is some state $\trnsstate \in \settrnsstate$ of the encoded transducer $T$ for which $\position_\inputexample(\stringlength{\inputexample}) = (\stringlength{\outputexample}, \trnsstate)$.
From Lemma~\ref{lem:example_invariant}, it follows that $\extendeddeltaTout(\initstate{T}, \inputexample) = \outputexample$ and
 $\extendeddeltaTstate(\initstate{T}, \inputexample) = \trnsstate $.
This is exactly what it means for $T(\inputexample) = \outputexample$.

For the backward direction, assume that for the transducer $T$ encoded in $\vardeltastate$, $\vardeltaout$,  and $\vardeltaoutlen$, we have $T(\inputexample) = \outputexample$.
Then, there exists a run $\initstate{T} \xrightarrow{a_0 / y_0} \trnsstate^1 
\xrightarrow{a_1 / y_1} \trnsstate^2
\ldots
\trnsstate^{n} \xrightarrow{a_n / y_n} \trnsstate^{n+1}$ such that $s = a_0\cdots a_n$ and $t = y_0 \cdots y_n $.
Using this run of $\inputexample$ on $T$, we can construct a definition of $\position_\inputexample$ such that all constraints in $\phi$ are satisfied.

\end{proof}

\subsection{Input-Output Types}
\label{sec:types_proofs}

Lemmas~\ref{lem:simulation_lemma} and \ref{lem:types_correctness} state that \rone the simulation relation $\simulationsym$ preserves the desired invariant, which is then used to prove that \rtwo $T$ satisfies the specification for input-output types.

\begin{lemma}
\label{lem:simulation_lemma}
Let $\phi$ be the set of input-output constraints from Equation~\ref{eq:types1} and Equation~\ref{eq:types2} for the types $\inputtype$ and $\outputtype$, and $T$ be a transducer encoded into the functions $\vardeltastate$, $\vardeltaout$, and $\vardeltaoutlen$.
If all constraints in $\phi$ are satisfied by $\vardeltastate$, $\vardeltaout$,  $\vardeltaoutlen$, and $\simulationsym$, 
then $\simulation{\srcstate}{\trnsstate}{\trgstate}$ if there exist strings $w,w'$ such that
    $\extendedtransition{\inputtype}(\initstate{\inputtype}, w) = q_{\inputtype}$, 
    $\extendeddeltaTstate(\initstate{T}, w) = \trnsstate$, 
    $\extendeddeltaTout(\initstate{T}, w) = w'$, 
    and $\extendedtransition{\outputtype}(\initstate{\outputtype}, w') = q_{\outputtype}$.
\end{lemma}

\begin{proof}
\sloppypar Assume that there exist strings $w, w'$ such that $\extendedtransition{\inputtype}(\initstate{\inputtype}, w) = \srcstate$,
$\extendeddeltaTstate(\initstate{T}, w) \allowbreak{}= \trnsstate$, 
$\extendeddeltaTout(\initstate{T}, w) = w'$, and
$\extendedtransition{\outputtype}(\initstate{\outputtype}, w') = \trgstate$.
By Constraint~\ref{eq:types1}, we already have that $\simulation{\initstate{\inputtype}}{\initstate{T}}{\initstate{\outputtype}}$.
Since $\extendeddeltaTout(\initstate{T}, w) = w'$, we know that there is an assignment to the functions $\vardeltastate$, $\vardeltaout$, and $\vardeltaoutlen$ that encodes a run $\initstate{T} \xrightarrow{a_0 / y_0} \trnsstate^1 
\xrightarrow{a_1 / y_1} \trnsstate^2
\ldots
\trnsstate^{n} \xrightarrow{a_n / y_n} \trnsstate^{n+1}$ such that $w = a_0\cdots a_n$ and $w' = y_0 \cdots y_n $.
It follows that we eventually reach $\simulation{\srcstate}{\trnsstate}{\trgstate}$ by the implication in Constraint~\ref{eq:types2}.

\end{proof}

\begin{lemma}
\label{lem:types_correctness}
Let $T$ be a transducer encoded into the functions $\vardeltastate$, $\vardeltaout$, and $\vardeltaoutlen$ and $\phi$ be the set of input-output type constraints of the types given in Equation~\ref{eq:types1}, Equation~\ref{eq:types2} and Equation~\ref{eq:types3}. 
Then $\vardeltastate$, $\vardeltaout$, and $\vardeltaoutlen$ can be extended to a model for $\phi$ if and only if $\hoaretriple{\inputtype}{T}{\outputtype}$.
\end{lemma}

\begin{proof}
For the forward direction, we assume that an assignment to $\vardeltastate$, $\vardeltaout$, $\vardeltaoutlen$ and $\simulationsym$ satisfies all constraints in $\phi$. 
Then, for every $\srcstate \in \finalstates{\inputtype}$, if we have that $\simulation{\srcstate}{\trnsstate}{\trgstate}$ for some $\trnsstate, \trgstate$, then it is the case that $\trgstate \in \finalstates{\outputtype}$.
From Constraint~\ref{eq:types2}, we know that for a string $w$ such that $\extendedtransition{\inputtype}(\initstate{\inputtype}, w) = q_{\inputtype}$ where $\srcstate \in \finalstates{\inputtype}$, it must be that $\simulation{\srcstate}{\trnsstate}{\trgstate}$ for at least some $\trnsstate$ and $\trgstate$.
From Lemma~\ref{lem:simulation_lemma}, it also follows that for such a string $w \in P$, there is also a string $w'$ such that $\extendeddeltaTout(\initstate{T}, w) = w'$ and $\extendedtransition{\outputtype}(\initstate{\outputtype}, w') = \trgstate$ such that $\trgstate \in \finalstates{\outputtype}$.
In other words, for every string $w \in P$, $T$ outputs a string $w' \in Q$ as desired.

For the backward direction, let $T$ be a transducer such that $\hoaretriple{\inputtype}{T}{\outputtype}$.
We can use the transition relation of $T$ to construct an assignment of $\vardeltastate$, $\vardeltaout$, $\vardeltaoutlen$ and $\simulationsym$ such that all constraints are satisfied.

\end{proof}

\subsection{Input-Output Distances}
\label{sec:distances_proofs}

Lemmas~\ref{lem:distancelemma} and \ref{lem:distance_correctness} show that a transducer $T$ that is a solution to the set of distance constraints is such that $\meaneddist{w}{T(w)} \leq \distancebound$ for all $w \in \DFAlanguage{\inputtype}$.

\begin{lemma}
Let $\phi$ be the set of input-output distance constraints from Equations~\ref{eq:distance1}, \ref{eq:distance2}, \ref{eq:distance4}, \ref{eq:distance5} and \ref{eq:distance6} for a mean edit distance of $\distancebound$, and $T$ be a transducer encoded into the functions $\vardeltastate$, $\vardeltaout$ and $\vardeltaoutlen$.
If all constraints in $\phi$ are satisfied by $\vardeltastate$, $\vardeltaout$,  $\vardeltaoutlen$, $\simulationsym$, $\varprogress$, and $\varweight$, then for all runs over $T$ of the form $
\initstate{T} \xrightarrow{a_0/y_0} \trnsstate^1 \ldots \trnsstate^{n-1} \xrightarrow{a_{n-1}/y_{n-1}} \trnsstate$, 
where $a_0\cdots a_{n-1} \in \DFAlanguage{P}$, it holds that:

\begin{displaymath}
\sum_{i = 0}^{n-1} \normalfont{\varweight(\trnsstate^i, a_i)}
\geq 0.
\end{displaymath}

\label{lem:distancelemma}
\end{lemma}

\begin{proof}
Assume that all constraints in $\phi$ are satisfied by $\vardeltastate$, $\vardeltaout$,  $\vardeltaoutlen$, $\simulationsym$, $\varprogress$ and $\varweight$.
Consider an arbitrary path over of the form $
\initstate{T} \xrightarrow{a_0/y_0} \trnsstate^1 \ldots  \trnsstate^{n-1} \xrightarrow{a_{n-1}/y_{n-1}} \trnsstate$, 
where $a_0a_1\cdots a_{n-1} \in \DFAlanguage{P}$.
By Constraint~\ref{eq:distance5}, we have that

\begin{displaymath}
\begin{split}
\energyfunc{\initstate{\inputtype}}{\initstate{T}}{\initstate{\outputtype}} & \geq
\energyfunc{\srcstate^1}{\trnsstate^1}{\trgstate^1} -
\varweight(\initstate{T}, a_0) \\
\ldots \ldots & \geq  
\energyfunc{\srcstate^{n}}{\trgstate^{n}}{\trnsstate^{n}} - 
\sum_{i = 0}^{n-1} \varweight(\trnsstate^i, a_i) \\
\end{split}
\end{displaymath}

Since $w \in P$, we have that $\srcstate^{n} \in \finalstates{\inputtype}$.
By Constraint~\ref{eq:distance6}, it follows that $\energyfunc{\srcstate^{n}}{\trgstate^{n}}{\trnsstate^{n}} \geq 0$.
Therefore, $\energyfunc{\initstate{\inputtype}}{\initstate{T}}{\initstate{\outputtype}} + 
\sum_{i = 0}^{n-1} \varweight(\trnsstate^i, a_i)
\geq 0$ as well.
Since $\energyfunc{\initstate{\inputtype}}{\initstate{T}}{\initstate{\outputtype}} = 0$ by Constraint~\ref{eq:distance4}, we have that \linebreak
$\sum_{i = 0}^{n-1} \varweight(\trnsstate^i, a_i)
\geq 0$.

\end{proof}

\begin{lemma}
\label{lem:distance_correctness}
Let $T$ be a transducer encoded into the functions $\vardeltastate$, $\vardeltaout$ and $\vardeltaoutlen$, and $\phi$ be the set of input-output distance Constraints~\ref{eq:distance1}, \ref{eq:distance2},  \ref{eq:distance4}, \ref{eq:distance5} and \ref{eq:distance6}.
If $\vardeltastate$, $\vardeltaout$ and $\vardeltaoutlen$ can be extended to a model for $\phi$, then for any string $w \in \DFAlanguage{\inputtype}$, $\meaneddist{w}{T(w)} \leq \distancebound$. 
\end{lemma}

\begin{proof}
Assume that an assignment to $\vardeltastate$, $\vardeltaout$, $\vardeltaoutlen$, $\simulationsym$, $\varprogress$ and $\varweight$ satisfy all constraints in $\phi$. 
From Lemma~\ref{lem:distancelemma} it follows that the run of any string $w \in \DFAlanguage{\inputtype}$ is such that $\sum_{i = 0}^{n-1} \varweight(\trnsstate^i, a_i)
\geq 0$, where $a_0\cdots a_{n-1} = w$ and $\trnsstate^0\cdots \trnsstate^{n-1}$ are the states in the run.
Since $\varweight(\trnsstate^i, a_i)$ is the difference between $\distancebound$ and $\vareddist(\trnsstate^i, a_i)$ at each transition, if the total difference is $\geq 0$, then it follows that  $\meaneddist{w}{T(w)} \leq \distancebound$.

\end{proof}

\section{Size of Constraint Encoding}
\label{sec:size_encoding}

\mypar{Input-Output Examples.}

There is one constraint of type~\ref{eq:examples1}, $n m |\settrgstate||\Sigma|$ constraints of type~\ref{eq:examples2}, and one constraint of type~\ref{eq:examples3} for one input-output example, where $n$ is the length of the input and $m$ is the length of the output.

A constraint of type~\ref{eq:examples1} involves 1 variable, a constraint of type~\ref{eq:examples2} involves $4 + l$ variables, and a constraint of type~\ref{eq:examples3} involves $|\settrnsstate|$ variables.

\mypar{Input-Output Types.}

There is one constraint of type~\ref{eq:types1}, $|\setsrcstate||\settrnsstate||\settrgstate||\Sigma|$ constraints of type~\ref{eq:types2}, and one constraint of type~\ref{eq:types3} for the input-output types $\inputtype$ and $\outputtype$.

A constraint of type~\ref{eq:types1} involves no variables, a constraint of type~\ref{eq:types2} involves $4 + 3l$ variables, and a constraint of type~\ref{eq:types3} involves $|\finalstates{\inputtype}| + |\settrgstate - \finalstates{\outputtype}| + |\settrnsstate|$ variables.

\mypar{Input-Output Distances.}

There are $|\settrnsstate||\Sigma|l$ constraints of type~\ref{eq:distance1} and \ref{eq:distance2}, one constraint of type~\ref{eq:distance4}, $|\setsrcstate||\settrnsstate||\settrgstate||\Sigma|$ constraints of type~\ref{eq:distance5}, and 1 constraint of type~\ref{eq:distance6} for a specified input-output distance.

Constraints of type~\ref{eq:distance1} and ~\ref{eq:distance2} involve up to $2 + l$ variables, a constraint of type~\ref{eq:distance4} involve no variables, a constraint of type~\ref{eq:distance5} involves $4 + 3l$ variables, and a constraint of type~\ref{eq:distance6} involves $|\finalstates{\inputtype}| + |\settrnsstate| + |\settrgstate|$ variables.

\section{Soundness of s-FT recovery algorithm}
\label{sec:sft_synthesis_soundness}

The following lemma shows that our s-FT recovery algorithm, described in Section~\ref{sec:symbolic_extension}, which we use to synthesize s-FTs is sound in the sense that finitizing the s-FT will result in an identical finite transducer.

\begin{lemma}
Let $T$ be an FT and let $T'$ be the corresponding s-FT obtained using the s-FT extraction algorithm.
Then we have that $T = \finite{T'}$ if we use the representative character $\witnessshort{\psi}$ for each minterm $\psi$. 
\end{lemma}
\begin{proof}
Without loss of generality, consider a transition $q \xrightarrow{c/y} q'$ where the input character is $c$ and the output character is $y$.
The proof generalizes for an output string with multiple characters.
Now, one of 3 cases apply for $c$:
\begin{enumerate}
    \item If $c = y$, then the output function $f$ that we choose is the identity function. Re-finitizing this transition results in the transition $q \xrightarrow{c/c} q'$, as expected.
    \item If $\mintermpred{c}$ and $\mintermpred{y}$ are single intervals $I_1, I_2$ of the same size, then $f = x + \textit{offset}(I_1, I_2)$. If we were to finitize the transition with the same witness for each minterm, then we get back $q \xrightarrow{c/y} q'$.
    \item Otherwise, we have a constant function $f = y$. Finitizing the transition again results in $q \xrightarrow{c/y} q'$.
\end{enumerate}
\end{proof}

\section{Constraints for synthesizing Transducers with Regular Lookahead}
\label{sec:lookahead_constraints}

In writing the lookahead constraints for a string $w$, we omit the subscript $w$ and just write $\varlookahead$ and $\lookahead$ for $\varlookahead_w$ and $\lookahead_w$.

\mypar{Input-Output Examples and Lookahead.} 
Constraint~\ref{eq:lookahead1} expresses how we compute the lookahead states for a string $w = a_0\cdots a_{n}$.
\begin{equation}
\label{eq:lookahead1}
    \varlookahead(n) = \initstate{\lookaheadaut} \wedge
    \bigwedge_{0 \leq i < n}
    \varlookahead(i) = \vardeltalookahead(\varlookahead(i + 1), a_{i+1})
\end{equation}

Constraint~\ref{eq:lookahead2} shows how the output configuration for input-output examples takes into account what the lookahead state is at any point.
Given an example $\examplemap$ such that $\inputexample=a_0\cdots a_n$ and $\outputexample = b_0\cdots b_m$, for every $0\leq i < n$, $0 \leq j \leq m$, $c\in\Sigma$, $\lookaheadstate \in \lookaheadaut$ and $q_T\in Q_T$:
\begin{equation}
\label{eq:lookahead2}
\begin{split}
     & \position_\inputexample(i) = (j, \trnsstate)
    \wedge
    a_i = c \wedge 
    \varlookahead(i) = \lookaheadstate
    \Rightarrow \\
    & 
    \quad[
    \bigwedge_{0 \leq z < l} (\vardeltaout(\trnsstate, (\lookaheadstate, c),z) {=} b_{j+z} \vee z\geq  \vardeltaoutlen(\trnsstate, (\lookaheadstate, c)) \wedge\\
    &
    \quad\position_\inputexample(i + 1) {=} (j + \vardeltaoutlen(\trnsstate, (\lookaheadstate, c)), \vardeltastate(\trnsstate, (\lookaheadstate, c)))
    ] 
\end{split}
\end{equation}

\mypar{Input-Output Types with Lookahead.} The simulation relation $\simulationsym$ needs to also consider where we are in $R$ during any possible run on $T$.
So, we extend $\simulationsym$ to include the states of $R$---i.e., $\lookaheadsimulationsym \subseteq \setsrcstate \times \settrnsstate \times \settrgstate \times \setlookaheadstate$.
Since we travel `backwards' over $\lookaheadaut$, while moving `forward' over $\inputtype$, $T$ and $Q$ in the simulation relation, the simulation relation uses the inverse transition function of $\lookaheadaut$, which we denote by $\invtransition{\lookaheadaut}$---i.e., $\invtransition{\lookaheadaut}(\lookaheadstate, c) = \{\lookaheadstate' \mid \lookaheadstate' \xrightarrow{c} \lookaheadstate\}$. 
We also introduce the corresponding uninterpreted function $\vardeltainverselookahead$, which is the inverse of $\vardeltalookahead$.
The desired invariant is that $\lookaheadsimulation{\srcstate}{\trnsstate}{\trgstate}{\lookaheadstate}$ if there exists a lookahead word $w_{\lookaheadaut} = (\lookaheadstate^0, a_0)\cdots (\lookaheadstate^n, a_n)$ and a string $w$ such that
$\extendedtransition{\inputtype}(\initstate{\inputtype}, a_0\cdots a_n) = q_{\inputtype}$, 
$\extendeddeltaTstate(\initstate{T}, w_{\lookaheadaut}) = \trnsstate$, $\extendeddeltaTout(\initstate{T}, w_{\lookaheadaut}) = w'$, 
$\extendedtransition{\outputtype}(\initstate{\outputtype}, w') = q_{\outputtype}$ and 
$\extendedinvtransition{R}(\lookaheadstate, w) = \initstate{R}$.
%

Our starting point in the simulation can be any state $\lookaheadstate \in R$.
This is because $R$ reads a string $w \in \inputtype$ backwards, so the run could start from any state in $\lookaheadaut$.
\begin{equation}
\label{eq:lookahead3}
    \bigvee_{\lookaheadstate \in R}
    \lookaheadsimulation{\initstate{\inputtype}}{\initstate{T}}{\initstate{\outputtype}}{\lookaheadstate}
\end{equation}

Then, to advance the simulation, we now use the inverse transition relation of $R$.
So, if we are at state $\lookaheadstate \in \setlookaheadstate$, then for each character $\achar \in \Sigma$, the simulation advances to all the states that we could have reached $\lookaheadstate$ from by reading $c$. For any combination of $\srcstate, \trnsstate, \trgstate, \lookaheadstate, \lookaheadstate'$ and $c \in \Sigma$, we use the free variables $c_0 \ldots, c_{l-1}$ and $\trgstate^0 \ldots, \trgstate^{l-1}$ to encode the following constraint:
\begin{equation} 
\label{eq:lookahead4}
\begin{split}
    & \lookaheadsimulation{\srcstate}{\trnsstate}{\trgstate}{\lookaheadstate} \wedge \vardeltainverselookahead(\lookaheadstate, c) = \lookaheadstate' \Rightarrow \\
    & \bigwedge_{0 \leq z \leq \outputbound}  (\vardeltaoutlen(\trnsstate, (\lookaheadstate, c)) = z \Rightarrow \\
    & \quad\quad [\bigwedge_{0 \leq x < z} \vardeltaout(\trnsstate, (\lookaheadstate, c), x)=c_x] \wedge \\
    & \quad\quad [\trgstate^0 = \trgstate \wedge \bigwedge_{1 \leq x < z} \trgstate^x =
    \vardeltaoutputtype(\trgstate^{x-1},c_{x-1})] \wedge \\
    & \quad\quad \lookaheadsimulation{\transition{\inputtype}(\srcstate, \mathit{c})}{\vardeltastate(\trnsstate, (\lookaheadstate, c))}{\trgstate^z}{\lookaheadstate'})
\end{split}
\end{equation}

Finally, constraint~\ref{eq:lookahead5} says that when we reach a final state $\srcstate \in \finalstates{\inputtype}$ and the initial state of $\lookaheadaut$ (since we are going backwards over $\lookaheadaut$), we cannot be in a non-final state of $\outputtype$.
\begin{equation}
\label{eq:lookahead5}
    \bigwedge_{\srcstate \in \finalstates{\inputtype}}
    \bigwedge_{\trgstate \notin \finalstates{\outputtype}}
    \neg \lookaheadsimulation{\srcstate}{\trnsstate}{\trgstate}{\initstate{R}}
\end{equation}

\section{Alternate Notions of Input-Output Distances}
\label{sec:distances}

It is easy to modify our constraints to support the following different notions of input-output distances.
For example, we can ask that for every input $s$, the output $T(s)$ can have edit distance at most $D$ from $s$. 
To enforce this distance, we change Constraints~\ref{eq:distance4} and \ref{eq:distance5} to the constraints:
\begin{equation}
\label{eq:distancebounded}
    \energyfunc{\initstate{\inputtype}}{\initstate{T}}{\initstate{\outputtype})} = D
\end{equation}
\begin{equation}
\label{eq:distancebounded2}
\begin{split}
    & \bigwedge_{0 \leq z < l}  (\vardeltaoutlen(\trnsstate, c) = z \Rightarrow  \\ 
    &~~ [\bigwedge_{0 \leq x < z} \vardeltaout(\trnsstate, c, x)=c_x] \wedge \\
    &~~ [\energyfunc{\transition{\inputtype}(\srcstate, c)} {\vardeltastate(\trnsstate, c)} {\extendedtransition{\outputtype}(\trgstate, c_0\cdots c_{z-1})} = \\
    &~~ \energyfunc{\srcstate}{\trnsstate}{\trgstate} -
    \vareddist(\trnsstate, c) ]
\end{split}
\end{equation}

\end{document}